\documentclass[11pt]{article}
\usepackage[utf8]{inputenc}	

\usepackage[dvipsnames]{xcolor}
\usepackage{amsmath,amsthm,amsfonts,amssymb,amscd}
\usepackage{sn-preamble-v2}

\usepackage[T1]{fontenc}
\usepackage{multirow,booktabs}
\usepackage{yfonts}
\usepackage{fullpage}
\usepackage{lastpage}
\usepackage{enumitem}
\usepackage{fancyhdr}
\usepackage{algorithm}
\usepackage{algpseudocode}
\usepackage{mathrsfs}
\usepackage{wrapfig}
\usepackage{setspace}

\DeclareMathAlphabet{\mathfrak}{U}{jkpmia}{m}{it}
\SetMathAlphabet{\mathfrak}{bold}{U}{jkpmia}{bx}{it}

\usepackage{setspace}

\usepackage{pdfpages}

\usepackage{calc}
\usepackage{multicol}
\usepackage{fontenc}
\usepackage{cancel}
\usepackage[retainorgcmds]{IEEEtrantools}

\usepackage{amsmath}
\usepackage{url}

\usepackage{bbm}

\usepackage{empheq}
\usepackage{framed}

\usepackage{xcolor}
\usepackage[most]{tcolorbox}

\usepackage{hyperref}
\hypersetup{
    colorlinks=true,
    linkcolor=blue,
    filecolor=magenta,      
    urlcolor=cyan
    }
\usepackage[nameinlink]{cleveref} \colorlet{shadecolor}{orange!15}

\newcommand{\Yes}{{\textsf{Yes}}}

\newcommand{\No}{{\textsf{No}}}

\newcommand{\rd}{{\sf reldist}}
\newcommand{\reldist}{{\sf reldist}}
\newcommand{\dist}{{\sf dist}}
\newcommand{\dtv}{\mathrm{d}_{\mathrm{TV}}}
\newcommand{\absdist}{\Delta}

\newcommand{\zos}{\{0,1,*\}}

\newcommand{\SAMP}{\mathrm{SAMP}}
\newcommand{\MQ}{\mathrm{MQ}}

\newcommand\Algphase[1]{%
\vskip0.05in
\Statex\hspace*{-\algorithmicindent}\textbf{#1}%
\vskip0.05in
}

\algdef{SE}[REPEATN]{RepeatN}{End}[1]{\algorithmicrepeat\ #1 \textbf{times}}{\algorithmicend}

\algdef{SE}[REPEATN]{RepeatN}{End}[1]{\algorithmicrepeat\ #1 \textbf{times}}{\algorithmicend}

\newcommand{\UNATE}{\mathsf{Unate}}

\def\Edge{\textsf{Edge}}

\def\KnownN{\hyperref[algo:Known-N]{Unateness-Tester-Known-N}}
\def\BiasedTest{\hyperref[algo: BiasedTest]{BiasedTest}}
\def\UnbiasedTest{\hyperref[algo: UnbiasedTest]{UnbiasedTest}}

\def\IterativeBias{\hyperref[algo: IterativeBias]{IterativeBias}}

\def\fixed{\mathsf{Fixed}}
\def\unfixed{\mathsf{Unfixed}}

\newcommand{\Dyes}{\calD_{\mathrm{yes}}}
\newcommand{\Dno}{\calD_{\mathrm{no}}}
\def\TT{{\bT}}

\def\ff{{\boldf}}
\def\Dy{\mathcal{D}_{\text{yes}}}
\def\Dn{\mathcal{D}_{\text{no}}}
\def\bGamma{\boldsymbol{\Gamma}}
\def\HH{{\bH}}
\def\hh{{\bh}} 
\def\Ey{\mathcal{E}_{\text{yes}}}
\def\En{\mathcal{E}_{\text{no}}}

\title{Relative-error unateness testing}
\author{
    Xi Chen \\ Columbia University
    \and Diptaksho Palit\thanks{D.P.\ was supported by the U.S. National Science Foundation under Grant No.\ 2022446.}\\ Boston University
    \and Kabir Peshawaria \\ Boston University
    \and William Pires \\ Columbia University
    \and Rocco A.\ Servedio\thanks{R.A.S.\ was supported by NSF Awards CCF-2211238 and CCF-2106429.}\\  Columbia University
    \and Yiding Zhang \\ Boston University
}
\begin{document}

\maketitle

\begin{abstract}
The model of \emph{relative-error property testing} of Boolean functions has been the subject of significant recent research effort \cite{CDHLNSY2024,rel-error-conj-DL,rel-error-junta}.
In this paper we consider the problem of relative-error testing an unknown and arbitrary $f: \zo^n \to \zo$ for the property of being a \emph{unate} function, i.e.~a function that is either monotone non-increasing or monotone non-decreasing in each of the $n$ input variables.

Our first result is a one-sided non-adaptive algorithm for this problem that makes $\tilde{O}(\log(N)/\eps)$ samples and queries, where $N=|f^{-1}(1)|$ is the number of satisfying assignments of the function that is being tested and the value of $N$ is given as an input parameter to the algorithm.
Building on this algorithm, we next give a one-sided adaptive algorithm for this problem that does not need to be given the value of $N$ and with high probability makes $\tilde{O}(\log(N)/\eps)$ samples and queries.

We also give lower bounds for both adaptive and non-adaptive two-sided algorithms that are given the value of $N$ up to a constant multiplicative factor. In the non-adaptive case, our lower bounds essentially match the complexity of the algorithm that we provide.

\end{abstract}

\pagenumbering{gobble}

\newpage

\pagenumbering{arabic}


\newpage

\section{Introduction}

A Boolean function $f: \zo^n \to \zo$ is said to be \emph{unate} if for each of the $n$ individual input variables, $f$ is either monotone non-decreasing or monotone non-increasing in that variable. 
(An alternate equivalent characterization is that $f$ is unate if there is an $n$-bit string $s \in \zo^n$ such that $g(x)=f(x \oplus s)$ is monotone, where $\oplus$ denotes bitwise exclusive-OR.)
Together with monotonicity, unateness has been intensively investigated from the vantage point of \emph{property testing}.  Indeed, unateness testing was already studied in the influential early work of Goldreich et al.~that set the stage for much contemporary research on Boolean function property testing \cite{GGLRS}, and it has since been studied in many other papers including \cite{CS13a,BMPR16,KhotShinkar16,BCPRS17,CWX17stoc,CWX17focs,LW18,CW19,BCPRS20,CDLNS24}.

For the case of Boolean-valued functions $f: \zo^n \to \zo$ over the Boolean hypercube, the state-of the-art upper and lower bounds on the query complexity of testing unateness in the standard property testing model\footnote{Recall that in the standard model, an $\eps$-testing algorithm is given query access to an unknown and arbitrary $f: \zo^n \to \zo$. The algorithm must accept with high probability if $f$ has the property of interest and must reject with high probability if for every function $g$ with the property, $f$ and $g$ disagree on at least an $\eps$-fraction of all points in $\zo^n$.} were given in \cite{BCPRS20,CWX17stoc,CW19}.
In \cite{BCPRS20} Baleshzar et al.~gave a $\tilde{O}(n/\eps)$-query one-sided\footnote{In the standard Boolean function property testing model, an algorithm has \emph{one-sided error} if it accepts with probability 1 whenever the function being tested has the desired property.  A testing algorithm which need not satisfy this guarantee is said to have \emph{two-sided error}.}, non-adaptive\footnote{A testing algorithm in the standard model is \emph{non-adaptive} if its choice of the $(i+1)$-st point to query does not depend on the results of the first $i$ queries, i.e.~such an algorithm may ``make all of its queries at once.''} $\eps$-tester for unateness of $n$-variable Boolean functions.
In \cite{CWX17stoc}, Chen, Waingarten and Xie gave a $\tilde{\Omega}(n)$ lower bound for any one-sided, non-adaptive algorithm that tests unateness to a certain constant error parameter $\eps>0$, and additionally gave a  $\tilde{\Omega}(n^{2/3})$ query lower bound for any two-sided adaptive algorithm.
In \cite{CW19} Chen and Waingarten gave a one-sided, adaptive algorithm that $\eps$-tests an arbitrary Boolean function for unateness using $\tilde{O}(n^{2/3}/\eps^2)$ queries.  Taken together these results provide essentially matching upper and lower bounds both for adaptive and non-adaptive unateness testing.

Given that unateness testing is quite well-understood in the standard property testing model, it is natural that researchers have turned to study unateness testing in more challenging extensions of the basic model.
For example, the works \cite{LW18,CDLNS24} have studied the problem of unateness in the \emph{tolerant} testing model.  

\medskip

\noindent {\bf This work: Relative-error unateness testing}.
In this paper, we consider the problem of unateness testing in a different extension of the standard Boolean function property testing model, namely the recently-introduced model of \emph{relative-error} property testing.
The relative-error model was introduced by Chen et al.~\cite{CDHLNSY2024}, motivated by the fact that the standard Boolean function property testing model is not well suited for testing \emph{sparse} Boolean functions (functions with very few satisfying assignments); this is simply because any function $f: \zo^n \to \zo$ with less than an $\eps$ fraction of satisfying assignments will trivially be $\eps$-close, in the standard model, to the constant-0 function.
To circumvent this, in the \emph{relative-error} Boolean function property testing model introduced by \cite{CDHLNSY2024} the distance between the function $f: \zo^n \to \zo$ that is being tested and a function $g: \zo^n \to \zo$ is defined to be
\[
\reldist(f,g) := {\frac {|f^{-1}(1) \hspace{0.05cm} \triangle \hspace{0.05cm} g^{-1}(1)|} {|f^{-1}(1)|}}.
\]
Hence relative distance is measured ``at the scale'' of the function $f$ that is being tested, i.e.~$|f^{-1}(1)|$, rather than at the ``absolute scale'' of $2^n =$ $ |\zo^n|$ that is used in the standard model; note that the relative-error way of measuring the distance between $f$ and $g$ makes sense even for very sparse functions $f$.
The model also allows the testing algorithm to obtain i.i.d.~uniform elements of $f^{-1}(1)$ by calling a ``random sample'' oracle; this is because if, as in the standard model, only black-box queries were allowed, then for very sparse $f$ a very large number of queries would be required to obtain any information about $f$ at all.  See \Cref{sec:rel-error} for a more detailed description of the relative-error model.

The original work of Chen et al.~\cite{CDHLNSY2024} showed, roughly speaking, that relative-error testing can never be easier than standard-model testing, and that for some (artificial) classes of functions it can be much harder, requiring $2^{\Omega(n)}$ queries even though these classes can be tested with $O(1)$ queries in the standard model (see \cite{CDHLNSY2024} for details).  It is natural then to ask, for various properties of interest that are known to have efficient standard-model testing algorithms, how hard they are to test in the relative-error model.
Several recent papers have shown that for some highly structured ``syntactic'' properties, such as being a Boolean conjunction or decision list \cite{rel-error-conj-DL} or being a $k$-junta \cite{rel-error-junta}, the query complexity of relative-error testing is essentially no larger than the complexity of standard-model testing.  A more subtle and interesting phenomenon, though, was evident in the results obtained for relative-error monotonicity testing in the original \cite{CDHLNSY2024} paper; we turn to this next.

The main results of \cite{CDHLNSY2024} for relative-error monotonicity testing were upper and lower bounds on the query complexity\footnote{For simplicity of exposition, we refer to the total number of random samples plus black-box queries made by a relative-error testing algorithm as its ``query complexity.''} of relative-error monotonicity testing that \emph{depended on the sparsity $N=|f^{-1}(1)|$ of the function $f$ being tested.}
The main positive result of \cite{CDHLNSY2024} is a one-sided adaptive algorithm which is an $\eps$-relative-error tester for monotonicity, and with high probability makes at most $O(\log(N)/\eps)$ queries, even when the value of $N$ is not known to the testing algorithm. Note that while this bound is $O(n/\eps)$ if $N=2^{\Theta(n)}$, it can be much smaller if $N$ is small.
On the lower bound side, \cite{CDHLNSY2024} showed that there is a constant $\eps=\Theta(1)$ such that for any constant $c<1$, any two-sided non-adaptive\footnote{See \Cref{sec:rel-error} for the formal definition of non-adaptive algorithms under the relative-error model.} $\eps$-relative-error tester for monotonicity must make $\tilde{\Omega}(\log N)$ queries even if $N$ is given to the algorithm provided that $N<2^{cn}$, and that any two-sided adaptive algorithm must make $\tilde{\Omega}((\log N)^{2/3})$ queries.  
Thus, unlike the ``syntactic'' properties of being a conjunction, decision list, or junta, for the property of monotonicity the sparsity of the function plays an essential role in the relative-error testing complexity.

In summary, given the relative-error testing results for monotonicity and other properties described above, the important role that unateness has played in standard-model property testing, and the fact that we understand standard-model unateness testing quite well (as described earlier), it is natural to consider the complexity of \emph{relative-error unateness testing}.  This is the subject of the current paper; we turn to describe our results below.

\subsection{Our results}

\noindent {\bf Upper bounds.}
We first give a one-sided, non-adaptive algorithm which takes as input the value of $N=|f^{-1}(1)|$:

\begin{theorem} [One-sided non-adaptive algorithm, known $N$, informal version of \Cref{thm:known-N}] \label{thm:firstpositive}
There is a one-sided, non-adaptive algorithm, \hyperref[algo:Known-N]{Unateness-Tester-Known-$N$}, for $\eps$-relative-error testing of unateness of an unknown and arbitrary function $f: \zo^n \to \zo$.  The \hyperref[algo:Known-N]{Unateness-Tester-Known-$N$} algorithm is given as input the value of $N$, and makes $\tilde{O}(\log(N)/\epsilon)$ samples and queries.
\end{theorem}

Our second and main algorithm does not need to be given the value of $N$:

\begin{theorem} [One-sided adaptive algorithm, unknown $N$, informal version of \Cref{thm:unknown-N}] \label{thm:secondpositive}
There is a one-sided, adaptive algorithm, \hyperref[algo:unknown]{Unateness-Tester-Unknown-$N$}, for $\eps$-relative-error testing of unateness of an unknown and arbitrary function $f: \zo^n \to \zo$.  With high probability the \hyperref[algo:unknown]{Unateness-Tester-Unknown-$N$} algorithm makes $\tilde{O}(\log(N)/\epsilon)$ samples and queries (and is not given as input the value of $N$).
\end{theorem}

\noindent {\bf Lower bounds.}  Our first lower bound holds for non-adaptive relative-error unateness testing algorithms, and applies even when the value of $N=|f^{-1}(1)|$ is given up to a constant multiplicative factor.  This result is as follows:

\begin{theorem}[Non-adaptive lower bound] \label{thm:firstlowerbound}
For any constant $\alpha_0<1$, there exists a constant $\eps_0>0$ such that any two-sided, non-adaptive algorithm for testing whether a function $f$ with $|f^{-1}(1)|=\Theta(N)$ for some given parameter $\smash{N\le 2^{\alpha_0 n}}$ is unate or has relative~distance~at least $\eps_0$ from unate must make $\tilde{\Omega}(\log N)$ queries.
\end{theorem}

\Cref{thm:firstlowerbound} shows that the complexity of the \hyperref[algo:Known-N]{Unateness-Tester-Known-$N$} algorithm is essentially best possible, even if two-sided error is allowed.

For our second lower bound, as we explain in \Cref{sec:techniques-lower}, an easy argument based on prior results shows that $\tilde{\Omega}((\log N)^{2/3})$ queries are needed for adaptive algorithms:

\begin{theorem}[Adaptive lower bound] \label{main:adaptivelowerbound}
For any constant $\alpha_0<1$, there exists a constant $\eps_0>0$ such that any two-sided, adaptive algorithm for testing whether a function $f$ with $|f^{-1}(1)|=\Theta(N)$ for some given parameter $\smash{N\le 2^{\alpha_0 n}}$ is unate or has relative~distance~at least $\eps_0$ from unate must make $\tilde{\Omega}((\log N)^{2/3})$ queries.
\end{theorem}
  
\subsection{Discussion}

Taken together, our positive and negative results give a fairly complete picture of the abilities and limitations of relative-error algorithms for testing unateness.
It is interesting to compare our results for relative-error unateness testing with the known results for unateness testing in the standard model \cite{CS16,CWX17stoc,CW19}.  Recall that in the standard model, essentially tight upper and lower bounds are known both for one-sided non-adaptive algorithms ($\tilde{\Theta}(n)$ queries) and for two-sided adaptive algorithms ($\tilde{\Theta}(n^{2/3})$ queries). In contrast, in the relative-error setting we give essentially tight upper and lower bounds of $\tilde{\Theta}(\log N)$ queries for one-sided non-adaptive algorithms (in fact our lower bound holds even for two-sided non-adaptive algorithms), but our upper bound for two-sided adaptive algorithms is $\tilde{O}(\log N)$ queries whereas our lower bound is $\tilde{\Omega}((\log N)^{2/3})$ queries.
Roughly speaking, the reason for this disparity with what can be proved in the standard model is as follows: The $\tilde{O}(n^{2/3})$-query standard-model unateness testing algorithm of \cite{CW19} crucially employs a novel use of binary search over long random paths in the Boolean hypercube $\zo^n$.  As demonstrated by our lower bound constructions, though, in the relative-error setting ``all of the action'' may take place over just two adjacent layers of the hypercube, which precludes an analysis using long paths in the hypercube. Thus, it is not clear how to use techniques inspired by \cite{CW19} to obtain an improved algorithmic result in the relative-error setting.  An interesting direction for future work is to close the gap left open by our results for adaptive relative-error unateness testing:  is the right answer closer to $(\log N)^{2/3},$ or to $\log N$?

It is also interesting to consider our result for relative-error unateness testing in the context of the \cite{CDHLNSY2024} results for relative-error monotonicity testing.  In both cases, the query complexity of relative-error testing depends on the sparsity $N=|f^{-1}(1)|$ of the function being tested; as discussed earlier, this is provably different from what happens for other properties such as conjunctions, decision lists, and juntas \cite{rel-error-conj-DL,rel-error-junta}. We note that those other properties are most naturally defined using a ``syntactic'' description of how the relevant functions can be represented in various models, while in contrast monotonicity and unateness are most naturally viewed as ``semantic'' properties which are defined in terms of the behavior  of the function's output values rather than in terms of a particular syntactic representation.  An interesting possible direction for future research would be to study other ``semantic'' Boolean function properties:  will it turn out that for them too, the complexity of relative-error testing is determined by the sparsity $N=|f^{-1}(1)|$?  Some possible properties to investigate in this context are \emph{union-closedness} and \emph{intersectingness}, which have been studied in the standard property testing model in recent work \cite{CDLNS23a,HP24}.


\section{Techniques}

\subsection{Upper bounds} \label{sec:techniques-upper}
We present two unateness testers in the relative-error model.
The first tester needs to be given the sparsity $N$ as an input, while the second does not.
As a starting point, we first introduce two related procedures from prior work that our algorithms will build on.
In each case, the brief description given below is a simplified version of the actual procedure given in the relevant paper.

\begin{enumerate}
    \item {\bf Relative-Error Monotonicity Testing, from \cite{CDHLNSY2024}} 
    This procedure repeatedly draws samples from the $\SAMP(f)$ oracle.
    For each sample $\bz \sim f^{-1}(1)$, it draws an index $\bi$ uniformly from $\{ i \in [n] \mid \bz_i = 0\}$, and queries $f(\bz \oplus e_{\bi})$.
    If the result is $0$, the procedure rejects.
    This test is one-sided, non-adaptive, and has query complexity $O(\log(N)/\eps)$; it is the inspiration for the \BiasedTest~subroutine in our paper.

    \item {\bf Standard-Model Unateness Testing, from \cite{BCPRS20}} 
    The idea of this procedure is to find both a strictly monotone and a strictly anti-monotone edge along the same coordinate $i$, which directly witnesses a  violation of unateness.
    The procedure can be thought of as interpolating between two extreme cases.
    First, when violations are concentrated entirely along one coordinate, we choose to sample many coordinates.
    However, for each coordinate that is sampled, we check only a few edges in that dimension.
    Second, when violations are spread equally across all coordinates, we choose to sample very few coordinates but for each chosen coordinate, we check more edges in that dimension.
    This test is one-sided, non-adaptive, and has query complexity $\tilde{O}(n/\eps)$; it is the inspiration for the \UnbiasedTest~subroutine in our paper.
\end{enumerate}

Of course, if $n = O(\log N)$, then we could directly use the standard-model unateness tester.
However, the interesting case for relative-error testing is when $n \gg \log N$.
Recall that a function $f$ is unate if and only if it is monotone along some direction $d \in \zo^n$.
To make use of the relative-error monotonicity tester, we would need to provide some direction $\tilde{d}$.
Unfortunately, it is not clear how to choose a suitable direction $\tilde{d}$;
indeed, if the relative-error monotonicity tester reports that the function $f$ is not monotone along some specific direction $\tilde{d}$, we cannot obviously conclude that $f$ is not unate --- perhaps $f$ is unate by virtue of being monotone along some other direction $d \neq \tilde{d}.$
Ideally, we would like to identify a direction $\tilde{d} \in \zo^n$ such that \emph{if $f$ is unate, then (with high probability) it is monotone along direction $\tilde{d}$.}

This goal is too optimistic because we cannot be confident in the setting of all $n$ coordinates.
However, we \emph{can} confidently set many coordinates of $\tilde{d}$, as we now explain.
For any unate function $f$, if the value of the $i$-th coordinate of a uniform sample drawn from $f^{-1}(1)$ is $1$ (resp.\ $0$) with probability $p_i > 0.5$, then $f$ must be monotone (resp.\ anti-monotone) along dimension $i$.
This notion is formalized in \Cref{cor:unate_Some_Empty_Edge}.
Thus, a natural approach is to draw some samples from $f^{-1}(1)$, use them to obtain an estimate $\widehat{p_i}$ of the above probabilities $p_i$ for each coordinate $i \in [n]$, and choose directions for the subset of coordinates $i$ where $\hat{p_i}$ is either close to 0 or close to 1 (see Phase~1 of the \KnownN\ algorithm).
We call such coordinates \emph{fixed} (since we fix the direction along these coordinates).
For these fixed coordinates, if we see a bichromatic edge along a coordinate $i$ that ``runs counter to $\widehat{p_i}$'' (for example, an anti-monotone edge along a coordinate $i$ such that $\widehat{p_i}$ is close to 1), we can be ``quite confident'' that $f$ is not unate.

At a high level, the strategy of our algorithm is to run (a variant of) the relative-error monotonicity tester along the fixed coordinates and to run (a variant of) the standard-model unateness tester along the unfixed coordinates (see Phase~3 of the \KnownN\ algorithm). 
Both testers are slightly modified to search for violations only on the coordinates provided.\footnote{This modification is easy, because both subroutines are edge-testers; we provide the relevant set of coordinates as input to the procedure, and we have the procedure only perform edge-tests along the coordinates provided in the input set.}

For this to succeed, there are two technical challenges that must be overcome.
First, in order to fix $n$ directions such that with constant probability we fix all of them correctly, $\Omega(\log n)$ samples are required.
For very sparse functions, such that $\log N = o(\log n)$, this is an unacceptably high number of samples.
Second, our modified standard-model unateness tester has query complexity which is quasilinear in the number of unfixed coordinates.
Thus, we must ensure that we fix all but $O(\log N)$ coordinates.

 \paragraph{The known-$N$ algorithm.}

When the sparsity $N$ is known, we solve both these issues by using an \emph{explicit truncation}.
As we argue in \Cref{lem:small-diameter}, if a function $f: \zo^n \to \zo$ is unate, then all $x, y \in f^{-1}(1)$ have Hamming distance at most $2 \log N$ from each other.
So prior to running the relative-error monotonicity and standard-model unateness testers, our procedure draws $O(\log N)$ samples from $f^{-1}(1)$ and rejects if any of the  received samples are too far (Hamming distance more than $2 \log N$) from the first sample $a \in \zo^n$ (see Phase~1 of the \KnownN\ algorithm).
Conditioned on not rejecting, we can now assume that we are estimating directions for a function $f$ such that all points in $f^{-1}(1)$ lie within a Hamming ball of radius $2 \log N$ around $a.$

We prove that for these \emph{truncated functions}, there can be at most $O(\log N)$ coordinates whose true bias $p_i$ (as defined above) is close to $0.5$.
In order to claim that $O(\log N)$ samples suffices to estimate the directions of the $n - O(\log N)$ remaining coordinates, our analysis distinguishes between \emph{trivial} and \emph{non-trivial} coordinates.
Trivial coordinates are ones for which $p_i$ is either exactly 0 or exactly 1.
Notably, our estimates for trivial coordinates cannot be wrong.
For truncated functions, all but $\poly(N)$ coordinates are trivial, so $O(\log N)$ samples suffices to correctly estimate the orientations of the non-trivial coordinates.

 \paragraph{The unknown-$N$ algorithm.}

When the sparsity $N$ is unknown, we have to do  more work.  For one thing, since $N$ is unknown we cannot simply draw $O(\log N)$ samples and use them (as was described earlier for Phase~1 of the \KnownN\ algorithm) to restrict our attention to a truncation of $f$ to a Hamming ball of radius $2\log N$. 
Despite this, it turns out that we are able to effectively work with a truncation to a Hamming ball of radius $O(N)$. 
This is accomplished using a procedure, which we call \hyperref[algo: CheckSamples]{CheckSamples} (see Phase~0 of the \hyperref[algo:unknown]{Unateness-Tester-Unknown-$N$} algorithm) which has only \emph{constant} query complexity, yet allows us to effectively ensure that we are working with a radius-$O(N)$ truncation of $f$. 
Since there are only $\poly(N)$ coordinates where such a truncation is non-trivial, it is possible to get good estimates $\widehat{p}_i$ for these coordinates using only $O(\log N)$ samples. 

While at this point only $O(\log N)$ samples would suffice to get good estimates $\widehat{p_i}$ for the required coordinates $i$, we still do not know $N$ so a naive approach will not work.  Instead, to estimate the $p_i$'s, we use a delicate iterative procedure called \hyperref[algo: IterativeBias]{IterativeBias} (see Phase~1 of \hyperref[algo:unknown]{Unateness-Tester-Unknown-$N$}). This procedure draws larger and larger sets of samples, until the estimates for the $p_i$'s converge. We are able to argue that this process halts after using roughly $\log N$ samples. Furthermore, we are able to argue that when this process halts, our estimates are ``good'' and moreover only $O(\log N)$ coordinates have estimates $\widehat{p_i}$ that are close to $0.5$.  

With these estimates in hand, we are able to get a better truncation of $f$. In particular, the next stage of our algorithm (see the \hyperref[alg: Preprocessing]{Preprocessing} procedure carried out in Phase~2 of \hyperref[algo:unknown]{Unateness-Tester-Unknown-$N$}) enables us to effectively restrict $f$ to a Hamming Ball of radius $O(\log(N))$. We can now run our modified standard-model tester and relative error monotonicity tester on this restriction (see Phase~3 of the \KnownN\ algorithm), while staying within our $O(\log(N))$ query complexity budget. 

We remark that the above discussion glosses over some subtleties which arise in our actual algorithms and analyses. One of these is the following: in order to obtain algorithms with the one-sided error guarantees that we achieve, it is not okay to reject if we see a violation of monotonicity along a fixed coordinate of the ``bias vector'' $\tilde{d}$ that our algorithm constructs.  This is because each fixed coordinate $i$ of $\tilde{d}$ is fixed  because of ``statistical evidence'' that the function is monotone (or anti-monotone) in direction $i$, but this statistical evidence does not constitute an incontrovertible witness of monotonicity (or anti-monotonicity) in that direction --- for a one-sided error testing algorithm, it is not enough to be ``quite confident'' as mentioned earlier.  So when the algorithm sees a violation of monotonicity (or anti-monotonicity) along a fixed coordinate of $\tilde{d}$, it must do some more work to actually obtain a violating edge for the other direction to establish for sure that the function is neither monotone nor anti-monotone along coordinate $i$. (This is why our algorithm uses the \hyperref[algo: ConfirmDirection]{ConfirmDirection}  procedure in various places.)

\paragraph{Overview of both algorithms.}
We close this subsection with a more detailed overview of the structure that is common to both our known-$N$ and unknown-$N$ algorithms.
\begin{itemize}
    \item In Phase~1, the algorithm derives a bias vector $\tilde{\bd}\in \zos^{n}$ based on the estimates $\{\hat{p_i}\}_{i\in [n]}$, while ensuring (with high probability) that the number of unfixed coordinates is at most $O(\log N)$ and for each fixed coordinate $i$, the true bias $p_i$ is indeed close to the value of $\tilde{\bd}_i\in \zo$. 
    \item In Phase~2, the algorithm determines the value of the radius parameter $\bM$, which informs the query complexity of \BiasedTest, and guarantees (with high probability) that $\bM=O(\log N)$.
    \item In Phase~3, the algorithm runs \BiasedTest\ and \UnbiasedTest\ given the parameters $\tilde{\bd}$ and $\bM$, where \BiasedTest\ takes as input $\bM$ and the fixed coordinates of $\tilde{\bd}$, and \UnbiasedTest\ takes as input the unfixed coordinates of $\tilde{\bd}$.
\end{itemize}

\subsection{Lower bounds} \label{sec:techniques-lower}

Our non-adaptive lower bound, \Cref{thm:firstlowerbound}, is obtained via a modification of the lower bound arguments from \cite{CDHLNSY2024}.
As in that work, our lower bound construction uses functions which are only nontrivial on two adjacent layers of the Boolean hypercube; we call these \emph{two-layer} functions
  (see \Cref{sec:two-layer}).
More precisely, a function $f$ is a two-layer function if it is only nontrivial
  on points in layers $3n/4$ and $3n/4+1$: every point with weight $<3n/4$ is
  set to $0$ and every point with weight $>3n/4+1$ is set to $1$.
(We use the constant $3/4$ just to make the presentation more concrete;
  it can be replaced by any constant strictly between $1/2$ and $1$).
  
Our ``yes-functions'' are the same as the ``yes-functions'' of \cite{CDHLNSY2024}, which are monotone (and hence unate) functions.  
However, we modify the ``no-functions'' of \cite{CDHLNSY2024} since the analysis of \cite{CDHLNSY2024} does not establish that the ``no-functions'' that were used in that work, to give lower bounds on relative-error monotonicity testing, are far from unate.
In \Cref{sec:lower} we give a detailed description of a new distribution of ``no-functions'' and prove that they are indeed suitably far from unate.  Given these no-functions, our lower bounds follow from relatively simple modifications of the \cite{CDHLNSY2024} lower bound arguments; we describe the necessary modifications in \Cref{sec:lower}.

Our adaptive lower bound (\Cref{main:adaptivelowerbound}) 
 is obtained by a simple reduction to the $\tilde{\Omega}(d^{2/3})$-query standard-model two-sided adaptive unateness testing lower bound given by \cite{CWX17stoc} for functions $f: \zo^d \to \zo$.  This is done by ``scaling down'' both the yes- and no- distributions used in that lower bound so that the functions involved are contained in a subcube of $\zo^n$ of dimension $d \approx \log N$; this can be done by simply defining the functions over $\zo^n$ to output 0 if any of the last $n-d$ input coordinates take value 1.  For such functions, the uniform random satisfying assignments that a relative-error testing algorithm may ask for can easily be obtained by a standard-model algorithm for unateness testing over $\zo^d$, by simply drawing uniform random assignments from $\{0,1\}^{d} \circ 0^{n-d}$ until a satisfying assignment is received.  It follows that the $\tilde{\Omega}(d^{2/3})$ standard-model lower bound of \cite{CWX17stoc} directly yields a $\tilde{\Omega}((\log N)^{2/3})$ lower bound for two-sided adaptive unateness testing with relative error.
(We mention that it is easily verified that in the \cite{CWX17stoc} construction, for both the distribution $\Dyes$ of yes-functions and the distribution $\Dno$ of no-functions, with very high probability the number of satisfying assignments of a random function drawn from either distribution is tightly concentrated around the same particular value.  This is why the lower bound of \Cref{main:adaptivelowerbound} holds even when $N$ is given to the testing algorithm.)

We remark that an alternate route to obtaining a $\tilde{\Omega}((\log N)^{2/3})$-query two-sided adaptive lower bound for relative-error unateness testing is to proceed in analogy with our proof of \Cref{thm:firstlowerbound}, i.e. to modify the ``no-functions'' used in the two-sided adaptive lower bound for relative-error monotonicity testing of \cite{CDHLNSY2024}.  Given suitably modified ``no-functions'' that are far from unate in relative error, the $\tilde{\Omega}((\log N)^{2/3})$ lower bound argument of \cite{CDHLNSY2024} for monotonicity adapts to give a $\tilde{\Omega}((\log N)^{2/3})$ lower bound for two-sided adaptive relative-error unateness testing, matching the parameters of \Cref{main:adaptivelowerbound}.  Since this construction and argument have many more details than the simple reduction sketched in the previous paragraph, though, we omit them.


\section{Preliminaries} \label{sec:prelim}

\subsection{Notation}

\paragraph{Strings.} For $x,b \in \zo^n$ we write $x \oplus b$ to denote the bitwise-XOR of $x$ and $b$, i.e.~the string in $\zo^n$ that has $x_i \oplus b_i$ as its $i$-th coordinate.  We also write $x^{(i)}$ to denote $x \oplus e_i$ for convenience. For $b \in \{0,1\}$ and $i \in [n]$ we write $x^{i \leftarrow b}$ to denote the $n$-bit string obtained from $x$ by fixing the $i$-th bit to $b$.
Given a partial string $d \in \zos^n$, we write $\fixed(d) := \{ i \in [n] \mid d_i \in \zo \}$ and $\unfixed(d) := \{ i \in [n] \mid d_i = *\}$ to denote the set of indices $i$ where $d_i$ is fixed or unfixed respectively. 
For partial strings $x, y \in \zos^n$, we denote their symmetric difference as
\[
    x \Delta y := \big\{ i \in \fixed(x) \cap \fixed(y) \mid x_i \neq y_i \big\}.
\]

\paragraph{Edges.} Given $i \in [n]$, the set of edges in the $i$-th direction is the set of (unordered) pairs 
$$\Edge_i=\big\{ \{x^{i \leftarrow 0},x^{i \leftarrow 1} \}:  x \in \{0,1\}^n\big\}.$$
We say an edge $\{x,y\}\in \Edge_i$ is \emph{monochromatic} if $f(x)=f(y)$ and \emph{bichromatic} otherwise. 

Given $f : \zo^n \to \zo$ and $i \in [n]$ we say an edge $\{x,y\} \in \Edge_i$ is \emph{strictly $1$-monotone} in $f$ if $f(x)=x_i$ and $ f(y)=y_i,$ and we say it is  \emph{strictly $0$-monotone} in $f$ if $f(x)=1-x_i$ and $f(y)=1-y_i.$
We let $\Edge^b_i(f)$  denote the set of all edges in $\Edge_i$ that are strictly $b$-monotone in $f$:
\begin{align*}
\Edge_{i}^b(f)&= \big\{ \set{x,x^{(i)}} \mid \set{x,x^{(i)}} \text{ is strictly $b$-monotone} \big\}.
\end{align*}

\paragraph{Distances.} 
We use $\absdist(x,y)$ to denote the Hamming distance between two strings $x,y \in \zos^n$, where $\absdist(x, y) = |x \Delta y|$.
We also use $\absdist(f, g)$ to denote the Hamming distance between two Boolean functions $f, g : \zo^n \to \zo$.
\[
    \absdist(f, g) := \Big| \big\{ x \in \zo^n \mid f(x) \neq g(x) \big\} \Big|.
\]
We write $\dist(f, g)$ to denote the normalized Hamming distance between two Boolean functions $f, g: \zo^n \to \zo$,
\[
    \dist(f,g) := \frac{\absdist(f, g)}{2^n}.
\]
For a class ${\cal C}$ of Boolean functions, we write $\dist(f, {\cal C})$ to denote $\min_{g \in {\cal C}}\hspace{0.05cm}\dist(f,g)$, the minimum distance between $f$ and any function $g \in {\cal C}.$
We sometimes say that \emph{$f$ is $\eps$-far from $g$ (from ${\cal C}$, respectively)} if $\dist(f,g) > \eps$ ($\dist(f,{\cal C})>\eps$, respectively).
In the standard model of property testing, an $\eps$-tester for a property is required to determine whether $f \in \calC$ or $\dist(f, \calC)\geq \epsilon$ (i.e., accept with probability $\geq 2/3$ if $f\in \calC$ and reject with probability $\geq 2/3$ if $\dist(f, \calP)\geq \epsilon$).

\subsection{Relative-error testing of Boolean functions} \label{sec:rel-error}

\noindent {\bf The relative-error property testing model.} As defined in \cite{CDHLNSY2024}, a \emph{relative-error} testing algorithm for a class ${\cal C}$ of Boolean functions has oracle access to $\MQ(f)$ and also has access to a $\SAMP(f)$ oracle which, when called, returns a uniform random element $\bx \sim f^{-1}(1)$.
A relative-error testing algorithm for ${\cal C}$ must output ``yes'' with high probability (say at least 2/3; this can be easily amplified) if $f \in {\cal C}$ and must output ``no'' with high probability (again, say at least 2/3) if $\reldist(f,{\cal C}) \geq \eps$, where
$$\reldist(f,{\cal C}):=\min_{g \in {\cal C}}\hspace{0.05cm}\reldist(f,g)\ \quad\text{and}\ \quad 
\reldist(f,g) := {\frac {|f^{-1}(1) \ \Delta \ g^{-1}(1)|}{|f^{-1}(1)|}}.
$$

We say that a relative-error testing algorithm is ``non-adaptive'' if after receiving the results of all of its calls to the sampling oracle, it makes one parallel round of queries to the black-box oracle (so these queries can depend on the result of the calls to the sampling oracle, but the choice of the $i$-th query point for the black-box oracle does not depend on the responses received to queries $1,\ldots, i-1$).

We sometimes say that \emph{$f$ is $\eps$-far from $g$ (from ${\cal C}$, respectively) in relative distance} if $\reldist(f,g) > \eps$ ($\reldist(f,{\cal C})>\eps$, respectively).
\noindent Throughout this writeup we denote
\begin{equation} \label{eq:N}
N := \big|f^{-1}(1)\big|
\end{equation}
where $f: \zo^n \to \zo$ is the function that is being tested for unateness.

We will use  the following simple ``approximate triangle inequality'' for relative distance:

\begin{lemma}[Approximate triangle inequality for relative distance, Lemma~5 of \cite{rel-error-conj-DL}]\label{lem: approx triangle ineq}
    Let $f,g,h:\{0,1\}^n \to \{0,1\}$ be such that $\reldist(f,g)\leq \epsilon$ and $\reldist(g,h)\leq \epsilon'$. Then $\reldist(f,h) \leq \epsilon+(1+\epsilon)\epsilon'$.
\end{lemma}

\subsection{Monotonicity and unateness}

For $i \in [n]$, we say a function $f:\{0,1\}^n \rightarrow \{0,1\}$ is \emph{$1$-monotone} in the $i$-th direction if every $\{x^{i \leftarrow 0},x^{i \leftarrow 1}\} \in \Edge_i$ satisfies $f(x^{i \leftarrow 0}) \le f(x^{i \leftarrow 1})$.
Similarly, we say a function $f$ is $0$-monotone in the $i$-th direction if every $\{x^{i \leftarrow 0},x^{i \leftarrow 1}\} \in \Edge_i$ satisfies $f(x^{i \leftarrow 0}) \ge f(x^{i \leftarrow 1})$.
We say $f$ is \emph{unate} if it is either $1$-monotone or $0$-monotone in the $i$-th direction for every $i\in [n]$.

\begin{definition}\label{def:dirVector}
Given $f:\{0,1\}^n\rightarrow \{0,1\}$, we define the \emph{bias vector $d^f \in \zos^n$ of $f$} as follows:
$$d^f_i=\begin{cases}
        1 &  \text{if } \Prx_{\bz \sim f^{-1}(1)}[\bz_i=1] > 0.6\\
        0 & \text{if } \Prx_{\bz \sim f^{-1}(1)}[\bz_i=0] > 0.6 \\
        * & \text{otherwise }
    \end{cases},\quad\text{for each $i\in [n]$.}$$
We will often consider vectors $\tilde{d} \in \{0,1,*\}^n$ 
  as estimations of $d^f$ and 
  compare the coordinates of such vectors to those of $d^f$. 
We say that $\tilde{d}\in \{0,1,*\}^n$ is \emph{consistent} with $d^f\in \{0,1,*\}^n$ if $\tilde{d}_i=d^f_i$ for every $i\in [n]$ with  $\smash{\tilde{d}}_i\ne *$. 
\end{definition}

The following is an immediate consequence of the definition of ${d}^f$:

\begin{observation}\label{lem:directions_probabilities}
    Let $f:\zo^n \to \zo$ be a function with bias vector $d^f\in \zos^n$. For any $i \in [n]$ with $d^f_i \neq *$, we have
    $$\Prx_{\bz \sim f^{-1}(1)}\big[\bz_i=d^f_i\big] \geq 1/2.$$
\end{observation}

We also have the following useful lemma, which roughly shows that it is easy to find a bichromatic edge in a biased direction given a uniformly random preimage of $f^{-1}(1)$:
\begin{lemma} \label{lem:BU}
    Let $f: \zo^n \to \zo$ be a function with bias vector $d^f \in \zos^n$. For any $i \in [n]$ with $d^f_i \neq *$ we have that
\begin{align}
\Prx_{\bz \sim f^{-1}(1)}\big[f(\bz^{(i)})=1-d^f_i] &\geq 1/5 \quad \text{and} \label{eq:rutabaga}\\
\Prx_{\bz \sim f^{-1}(1)}\big[\{\bz^{i \leftarrow 0}, \bz^{i \leftarrow 1}\} \in \Edge_{i}^{d_i^f}(f)] & \geq 1/5.
\label{eq:persimmon}
\end{align}
\end{lemma}
\begin{proof}
    Without loss of generality assume $d_i^f=1$ (the case of $d_i^f=0$ is identical). For $b \in \zo$, let $Z^b=\{z \in f^{-1}(1) : z_i =1\}$. Since $d^f_i=1$, we must have that  $|Z^1| > 0.6N$, so $|Z^0| < 0.4N$,
    and hence we have that
    $$\abs{\{z \in Z^1 : z^{(i)} \in Z^0\}} \leq 0.4N.$$
    Hence, there are at least $0.2N$ many $z \in f^{-1}(1)$ with $z_i=1, f(z^{(i)})=0$. Whenever $\bz \sim f^{-1}(1)$ is such a $z$, we have $f(\bz^{(i)})=0$, which gives us \Cref{eq:rutabaga}; and we moreover have $\{\bz^{i \leftarrow 0}, \bz^{i \leftarrow 1}\} \in \Edge_{i}^{1}(f)$, which implies
    $$\Prx_{\bz \sim f^{-1}(1)}[\{\bz, \bz^{(i)}\}\in \Edge_{i}^{1}(f)] \geq 1/5,$$
    giving us \Cref{eq:persimmon}.
\end{proof}

Our algorithms will reject if they find an $i \in [n]$ along with a pair of edges $e \in \Edge_i^0(f)$ and $e'\in \Edge_i^1(f)$. It is easy to see that when $f$ is unate, this can never happen: 
\begin{observation}\label{cor:unate_Some_Empty_Edge}
If $f: \zo^n \to \zo$ is unate, then for any $i \in [n]$ we have
    $$\Edge_{i}^{1-d^f_i}(f)=\emptyset.$$
\end{observation}

We will need the following theorem which relates the distance to unateness and the number of edges violating $b$-monotonicity:
\begin{theorem}[Theorem 1.3 of \cite{CS16}]\label{thm:CS16} 
If $f:\zo^n \to \zo$
 is $\epsilon$-far from unate, then   we have  $$\sum_{i\in [n]} \min\left(\Big|\Edge_{i}^{0}(f),\Edge_{i}^{1}(f)\Big|\right) \geq {\frac {\eps} {8}} \cdot 2^{n}. $$ 
\end{theorem} 

Recalling that $N=f^{-1}(1)$, where $f$ is the function being tested for unateness we have:

\begin{corollary}\label{corollary:CS16} 

If $f:\zo^n \to \zo$ is $\epsilon$-far from unate in relative distance, then we have  $$\sum_{i\in [n]} \left|\Edge_{i}^{1-d_i}(f)\right| \geq {\frac {\eps N} {8}},\quad\text{for every vector $d\in \{0,1\}^n$.}$$ 
\end{corollary} 

We will also use the following simple lemma, which shows that if $f$ is unate then the distance between any two points in $f^{-1}(1)$ is at most $O(\log N)$:
\begin{lemma} \label{lem:small-diameter}
    If $f:\zo^n \to \zo$ is unate, then $\absdist(z,z')\leq 2 \log N $
    for all $z,z' \in f^{-1}(1)$.
\end{lemma}
\begin{proof}
 Assume for contradiction there exists $z,z' \in f^{-1}(1)$ with $\absdist(z,z') > 2\log(N)$. Since $f$ is unate, there exists $r \in \zo^n$ such that $g(x)=f(x \oplus r)$ is monotone. 
 
 Let $x=z\oplus r$, $x'=z' \oplus r$. We have $g(x)=g(x')$. 
 Because $\absdist(z,z') = \absdist(x,x') > 2\log(N)$, either $\absdist(x,1^n)>\log(N)$ or $\absdist (x',1^n)>\log(N)$. 
 Without loss of generality, assume $\absdist(x,1^n) > \log(N)$. 
 Since $g$ is monotone, every point $y \succeq x$ must have $g(y)=1$. 
 But there are more than $2^{\log(N)}=N$ such points $y$, contradicting the fact that $|g^{-1}(1)|=|f^{-1}(1)|=N$. 
\end{proof}

\subsection{Truncated functions}

Given a positive integer $r$ and $\tilde{d}\in \{0,1,*\}^n$, we say a function $h:\{0,1\}^n\rightarrow \{0,1\}$ is $(r,\tilde{d})$-\emph{truncated} if every $x\in h^{-1}(1)$ satisfies $\absdist(x, \tilde{d}) \leq r$.

Given $M$ and $\tilde{d}$, an $(r,\tilde{d})$-truncated function $h= f_{r, \tilde{d}}$ is naturally obtained from $f$, the function being tested, as follows:
\begin{equation}\label{eq:bounded}
f_{r,\tilde{d}}(x):=\begin{cases}
        f(x) & \text{if $\absdist(x, \tilde{d}) \leq r$,} \\
        0 & \text{otherwise}.
    \end{cases}
\end{equation}

\begin{observation}
    \label{clm:unatefn-equals-truncation}
    Directly from \Cref{lem:small-diameter}, we observe that for unate function $f: \zo^n \to \zo$, input $a \in f^{-1}(1)$, and radius $r \geq 2\log N$ where $N := |f^{-1}(1)|$, the function $f$ is equivalent to its $(r, a)$-truncation, i.e. $f \equiv f_{r, a}$.
\end{observation}

\begin{lemma} \label{lem:loss-from-truncation}
    \label{lem:close_trunc_biases_are_similar_to_fn}
    For any Boolean function $f: \zo^n \to \zo$ with radius $0 \leq r \leq n$ and input $a \in \zo^n$ such that the truncation $h := f_{r, a}$ satisfies $\rd(f, h) \leq 0.05$, the following is true for all indices $i \in [n]$:
    $$\left| \Prx_{\bz \sim f^{-1}(1)}[\bz_i=1] - \Prx_{\bz \sim h^{-1}(1)}[\bz_i=1]\right| \leq 0.05.$$
\end{lemma}
\begin{proof}
    We instead prove a stronger claim.
    For functions $f, h$ where $h^{-1}(1) \subseteq f^{-1}(1)$ and any parameter $\eps \in (0,1)$, if $\rd(f, h) = \eps$ then writing $\bu_f$ ($\bu_h$, respectively) for a random variable that is uniform random over $f^{-1}(1)$ (over $h^{-1}(1)$, respectively), we have $\dtv(\bu_f,\bu_h) = \eps$, where $\dtv$ denotes the total variation distance.
    
    By definition of $\rd(f, h)$ and using $h^{-1}(1) \subseteq f^{-1}(1)$, we see that $|h^{-1}(1)| = (1 - \eps)N$. So
    \begin{align*}
        \dtv(\bu_f,\bu_h)
        &= \frac{1}{2} \left( \sum_{x \in f^{-1}(1) \setminus h^{-1}(1)} \Prx_{\bz \sim f^{-1}(1)}[\bz = x] + \sum_{x \in h^{-1}(1)} \left(\Prx_{\bz \sim h^{-1}(1)}[\bz = x] - \Prx_{\bz \sim f^{-1}(1)}[\bz = x]\right) \right) \\
        &= \frac{1}{2} \left( \frac{N - |h^{-1}(1)|}{N} + 1 - \frac{|h^{-1}(1)|}{N} \right) \\
        &= \frac{1}{2} \left( \eps + 1 - (1-\eps) \right) \\
        &= \eps. \qedhere
    \end{align*}
\end{proof}

\begin{definition}
\label{defn:nontrivial-coordinates}
    Fix Boolean function $f: \zo^n \to \zo$.
    We say that $i \in [n]$ is a \emph{nontrivial coordinate} if
    $$\Prx_{\bz \sim f^{-1}(1)}[\bz_i = 1] \not \in \zo.$$
\end{definition}

\begin{claim}\label{claim:trunc-fns-nontrivial-coord-bound}
Fix Boolean function $f: \zo^n \to \zo$, input $a \in f^{-1}(1)$, and radius $r \geq 2 \log N$ where $N := |f^{-1}(1)|$. 
We have $f_{r,a}^{-1}(1) \subseteq f^{-1}(1)$ and the number of nontrivial coordinates of $f_{r, a}$ is at most $rN$.
\end{claim}

\begin{proof}
The statement that $\smash{f_{r,a}^{-1}(1) \subseteq f^{-1}(1)}$ follows directly from the definition of truncation.

Given $z \in f_{r,a}^{-1}(1)$, there are at most $r$ coordinates where $z_i \neq a_i$. Furthermore, since $\smash{f_{r,a}^{-1}(1) \subseteq f^{-1}(1)}$ we have $\smash{|f_{r,a}^{-1}(1)| \leq N}$. 
So there can be at most $rN$ coordinates $i \in [n]$ for which there exists some $z \in f_{N,a}^{-1}(1)$ with $z_i \neq a_i$.
\end{proof}

\subsection{Setup for Algorithms and Analysis}

We present a key lemma that is central to the analysis of our main algorithms.
The lemma roughly states that given any function $h: \zo^n \to \zo$ that is far from unate (in $\rd$), there are many violations on at least one of two index sets which partition $[n]$.

\begin{claim}\label{clm:h_EQ}
     Let $\tilde{d}$ be an arbitrary vector in $\zos^n$ and let $h: \zo^n \to \zo$ be an arbitrary function. If $\reldist(h,\UNATE) \geq \eps/2$, then at least one of \Cref{eq:many-viols-on-fixed-idxs} or \Cref{eq:many-viols-on-unfixed-idxs} holds:   
     \begin{align} \sum_{i \in \fixed(\tilde{d})} \Big|\Edge_{i}^{1-\tilde{d}_i}(h)\Big|  &\geq \frac{\eps}{32}\cdot |h^{-1}(1)|, \label{eq:many-viols-on-fixed-idxs}\\
     \sum_{i \in \unfixed(\tilde{d})} \min \left(\Big|\Edge_i^0(h)\Big|,\Big|\Edge_{i}^1(h)\Big|\right) &\geq \frac{\eps}{32}\cdot |h^{-1}(1)|. \label{eq:many-viols-on-unfixed-idxs}
     \end{align}

\end{claim}

\begin{proof} 
Since $\reldist(h,\UNATE) \geq \eps/2$, we have $$\dist(h,\UNATE) \geq \frac{\epsilon  |h^{-1}(1)|}{ 2^{n+1}}.$$
By \Cref{thm:CS16}, 
$$\sum_{i\in [n]} \min \left(\Big|\Edge_{i}^{0}(h)\Big|,\Big|\Edge_{i}^{1}(h)\Big|\right) \geq \frac{\epsilon}{16}\cdot |h^{-1}(1)|.$$ 
Since $\fixed(\tilde{d})$ and $\unfixed(\tilde{d})$ partition $[n]$, it must be that at least one of \Cref{eq:many-viols-on-fixed-idxs,eq:many-viols-on-unfixed-idxs} holds.
\end{proof}

\begin{lemma}\label{lem:ftohtof}
    Assume $\rd(f,\UNATE) \geq \epsilon$. Let $M \geq 0$ and let $\tilde{d} \in \zos^n$.
    If $\rd(f,f_{M, \tilde{d}}) \leq \epsilon/10$ then at least one of 
\Cref{eq:many-viols-on-fixed-idxs_f_Md} or \Cref{eq:many-viols-on-unfixed-idxs_f} holds:   
     \begin{align}  
\sum_{i \in \fixed(\tilde{d})} \left| \Edge^{1-\tilde{d}_i}_i \left( f_{M, \tilde{d}} \right) \right|
    \ge \frac{\eps}{64} \cdot |f^{-1}(1)|,
    \label{eq:many-viols-on-fixed-idxs_f_Md}\\
    \sum_{i \in \unfixed(\tilde{d})}\min \left(\Big|\Edge_i^0(f)\Big|,\Big|\Edge_i^1(f)\Big|\right) \geq \frac{\eps}{64}\cdot |f^{-1}(1)|,
     \label{eq:many-viols-on-unfixed-idxs_f}
     \end{align}
\end{lemma}

\begin{proof}
For brevity, denote $f_{M, \tilde{d}}$ by $h$. Since $\reldist(f,h) \leq \epsilon/10$, we must have $|h^{-1}(1)| \geq |f^{-1}(1)|/2$. Furthermore by \Cref{lem: approx triangle ineq}, we have $\reldist(h, \UNATE) \geq \epsilon/2$. So, by \Cref{clm:h_EQ} at least one of \Cref{eq:many-viols-on-fixed-idxs} or \Cref{eq:many-viols-on-unfixed-idxs} holds for $h$.

First, consider the case where \Cref{eq:many-viols-on-fixed-idxs} holds.
Then \Cref{eq:many-viols-on-fixed-idxs_f_Md} holds as follows.
$$\sum_{i \in \fixed(\tilde{d})} \Big|\Edge_i^{1-\tilde{d}_i}(h)\Big| \geq \frac{\epsilon}{32} \cdot |h^{-1}(1)| \geq \frac{\epsilon}{64} \cdot |f^{-1}(1)|.$$ 

Now, consider the case where \Cref{eq:many-viols-on-unfixed-idxs} holds. We will show \Cref{eq:many-viols-on-unfixed-idxs_f} holds for $f$. Fix $i \in \unfixed(\tilde{d})$ and $b \in \zo$. Consider any edge $\{x,x^{(i)}\} \in \Edge_i^b(h)$; we claim that $\{x,x^{(i)}\} \in \Edge_i^b(f)$. From this it follows that: 
$$\sum_{i \in \unfixed(\tilde{d})}\min \left(\Big|\Edge_i^0(f)\Big|,\Big|\Edge_i^1(f)\Big|\right) 
\geq  \sum_{i \in \unfixed(\tilde{d})}\min \left(\Big|\Edge_i^0(h)\Big|,\Big|\Edge_i^1(h)\Big|\right) 
\geq \frac{\epsilon}{64} \cdot |h^{-1}(1)|.$$ 

To argue that  $\{x,x^{(i)}\} \in \Edge_i^b(f)$, without loss of generality assume $b=0$. And let $\{x,x^{(i)}\} \in \Edge_i^b(h)$ where $x_i=0$, $h(x)=1$ and $h(x^{(i)})=0$. Since $i \in \unfixed(\tilde{d})$,  we must have that 
$\absdist(x, \tilde{d}) = \absdist(x^{(i)}, \tilde{d}).$ 
Since $h(x)=1$, these distances must be less than $M$. 
Hence $f(x)=h(x)=1$ and $f(x^{(i)})=h(x^{(i)})=0$, implying $\{x,x^{(i)}\} \in \Edge_i^0(f)$.
\end{proof}


\section{Useful Subroutines} \label{sec:useful-subroutines}

\def\nil{\textsf{nil}}

\subsection{Testing in the biased and unbiased cases}\label{sec:mono-unate}

\Cref{theo:mono} and \Cref{thm:correctnessOfUnbiasedTest} establish performance guarantees for some useful subroutines, \hyperref[algo: BiasedTest]{BiasedTest} and \hyperref[algo: UnbiasedTest]{UnbiasedTest}, which are given in \Cref{algo: BiasedTest} and \Cref{algo: UnbiasedTest} respectively in \Cref{appendix:testing}.
These subroutines are used both in \hyperref[algo:Known-N]{Unateness-Tester-Known-$N$} and in \hyperref[algo:unknown]{Unateness-Tester-Unknown-$N$}.
The proofs of \Cref{theo:mono} and \Cref{thm:correctnessOfUnbiasedTest} can be found in \Cref{appendix:testing}.

\begin{restatable}{theorem}{mono}\label{theo:mono}
\hyperref[algo: BiasedTest]{BiasedTest} takes as input oracles $\MQ(f)$ and $\SAMP(f)$ of some $f:\{0,1\}^n\rightarrow \{0,1\}$, a positive integer $M$, a distance parameter $\eps'$ and 
  a vector $\tilde{d}\in \{0,1,*\}^n$.
It makes $O(M/\eps')$ queries to $\MQ(f)$ and $\SAMP(f)$ and rejects only when it has found two edges in $\Edge_i$ that violate unateness of $f$ along some coordinate $i$ with $\tilde{d}_i \in \zo$. 
Furthermore, when $\tilde{d}$ is consistent with $d^f$ and  
$$
\sum_{i \in \fixed(\tilde{d})}\left| \Edge^{1-\tilde{d}_i}_i \left( f_{M, \tilde{d}} \right) \right|\ge \eps' |f^{-1}(1)|,
$$
\hyperref[algo: BiasedTest]{\text{BiasedTest}}$(\MQ(f),\SAMP(f),M,\eps',\tilde{d})$ rejects with probability at least $0.99$.
\end{restatable} 

\begin{restatable}{theorem}{unatethm}\label{thm:correctnessOfUnbiasedTest}
     \hyperref[algo: UnbiasedTest]{UnbiasedTest}
     takes as input oracles $\MQ(f)$ and $\SAMP(f)$ of some $f:\{0,1\}^n\rightarrow \{0,1\}$, a distance parameter $\eps'$ and a vector $\tilde{d} \in \{0,1,*\}.$  
It makes $\tilde{O}(\frac{|\unfixed(\tilde{d})|}{\epsilon'})$ queries to $\MQ(f)$ and $\SAMP(f)$ and rejects only
  when it has found two edges in $\Edge_i$ that violate the unateness of $f$ along some coordinate $i\in \unfixed(\tilde{d})$.
Furthermore, when 
$$\sum_{i \in \unfixed(\tilde{d})} \min\left(\left|\Edge_{i}^0(f)\right|
    , \left|\Edge_{i}^1(f)\right|\right) \geq \epsilon'|f^{-1}(1)|,$$  
    \hyperref[algo: UnbiasedTest]{UnbiasedTest}$(\MQ(f), \SAMP(f), \epsilon', \tilde{d})$ 
    rejects with probability at least $0.99$.
\end{restatable}

\subsection{ConfirmDirection}

The \hyperref[algo: ConfirmDirection]{ConfirmDirection} subroutine will be used both in \hyperref[algo: BiasedTest]{\text{BiasedTest}} and in a preprocessing routine that is employed in our ``unknown $N$'' algorithm in \Cref{sec:unknown-N}.
It is given as input a coordinate $i \in [n]$ and a proposed orientation $b \in \zo$; it makes constantly many oracle calls, and returns $\Yes$ if it succeeds in finding an edge in $\Edge_i^b$ with the proposed orientation.

\begin{algorithm}[H]
\caption{ConfirmDirection} \label{algo: ConfirmDirection}
\vspace{0.15cm}
 \textbf{Input: } $\MQ(f)$ and $\SAMP(f)$ of some function $f$,  $i\in [n]$ and $b \in \zo$. \\
  \textbf{Output: } $\Yes$ if the algorithm found an edge in $\Edge_i^b(f)$ and $\No$ otherwsise. \\
 \begin{tikzpicture}
\draw [thick,dash dot] (0,1) -- (16.5,1);
\end{tikzpicture}
\begin{algorithmic}[1]
\RepeatN{25}
\If{$\set{\bz, \bz^{(i)}} \in \Edge_i^b$}
\State Return $\Yes$.
\EndIf
 \End
 \State Return $\No$.
\end{algorithmic}
\end{algorithm}

\begin{lemma}\label{lem:ConfirmDirection}
\hyperref[algo: ConfirmDirection]{ConfirmDirection} takes as input $\MQ(f), \SAMP(f)$ for some function $f$ and two paramaters $i \in [n], b \in \zo$.  It makes a constant number of queries to $\SAMP(f)$ and $\MQ(f)$ and returns \Yes ~if and only if it found an edge in $\Edge_i^b(f)$. Furthermore if $d^{f}_i=b$, then \hyperref[algo: ConfirmDirection]{ConfirmDirection}($\MQ(f), \SAMP(f), i, b)$ returns \Yes~with probability at least $0.9$. 
\end{lemma}
\begin{proof}
Clearly, the algorithm uses $50$ oracle calls and returns $\Yes$ if and only if it found an edge in $\Edge_i^b(f)$. It remains to prove this happens with high probability whenever $b=d^f_i$. By \Cref{lem:BU}, we have
$$\Prx_{\bz \sim f^{-1}(1)}\big[\set{\bz, \bz^{(i)}} \in \Edge_{i}^{d_i^f}(f)] \geq 1/5.$$

Hence, if $b=d^f_i$, each iteration of the for loop has a $1/5$ chance of returning $\Yes$. So, we have that \hyperref[algo: ConfirmDirection]{ConfirmDirection} returns $\Yes$ with probability at least $1-0.8^{25} \geq 0.9$ whenever $b=\tilde{d}_i$. 
\end{proof}

\begin{remark}\label{rem:confirm-dir-nonadaptive}
    \hyperref[algo: ConfirmDirection]{ConfirmDirection} can easily be made non-adaptive by first taking all samples, then making the corresponding membership queries, and finally checking the $\Yes$ condition.
\end{remark}

\subsection{Subroutine \hyperref[algo: CheckSamples]{CheckSamples}}\label{sec:analysis-CheckSamples}

The subroutine \hyperref[algo: CheckSamples]{CheckSamples} will be of critical importance for the tester that does not know the sparsity $N$ \textit{a priori}.
As discussed in the technical overview, critical parts of the analysis to come require considering the properties of a \emph{truncation} of the tested function $f$, rather than $f$ itself.
We remind the reader that as observed in \Cref{lem:small-diameter}, unate functions $f$ are unchanged by performing any truncation of radius at least $2\log N$ around any $a \in f^{-1}(1)$.
Thus, our analysis about a truncation of function $f$ is most useful when we can say that we can reject (with high probability) when presented with evidence that $f$ differs from the analyzed truncation.

This is accomplished in \Cref{algo: CheckSamples}. 
Given $a\in f^{-1}(1)$ and a set $S$ of points in $f^{-1}(1)$, we check if $S$ satisfies that $\absdist(a,z)\le 2N$ for every $z\in S$. (As discussed earlier,  this condition can never be violated when 
    $f$ is unate by \Cref{lem:small-diameter}.) 
The performance guarantees of \hyperref[algo: CheckSamples]{CheckSamples}\ are given in the following lemma:

\begin{algorithm}[t!]
\caption{CheckSamples} \label{algo: CheckSamples}
 \textbf{Input: }$\MQ(f)$ of some function $f$, $a\in f^{-1}(1)$, a set $S$ of points in $f^{-1}(1)$ and a parameter $\delta'$. \\
  \textbf{Output: }Either accept or reject.\\
 \begin{tikzpicture}
\draw [thick,dash dot] (0,1) -- (16.5,1);
\end{tikzpicture}
\begin{algorithmic}[1] 
\State
Breaking ties arbitrarily, let 
  $z$ be a point in $S$ that maximizes $\absdist(a, z)$ 
\RepeatN{ $\lceil\log(1/\delta')\rceil$ }
        \State  Draw $\bi \sim {a \Delta z}$ uniformly at random and  query $f(a^{(\bi)})$ and $f(z^{(\bi)})$ using $\MQ(f)$. 
        \State   Halt and reject if  {$f(a^{(\bi)})=f(z^{(\bi)})=0$}.
\End
\State Return. \Comment{The samples passed the test.}
 \end{algorithmic}
\end{algorithm}

\begin{lemma}\label{lem:CheckSamplesNoRej}
 \hyperref[algo: CheckSamples]{CheckSamples}
 takes as input $\MQ(f)$  of a function $f$, a point $a\in f^{-1}(1)$, a set $S$ of points in $f^{-1}(1)$, and  a parameter $\delta'$.
 It makes $O( \log (1/\delta'))$ queries, either rejects or returns, and satisfies the following performance guarantees:
\begin{flushleft}\begin{enumerate}
\item For any $f$, if $S$ contains at least one point $z\in S$ with $\absdist(a,z)>2N$, then  \hyperref[algo: CheckSamples]{CheckSamples} rejects with probability at least $1-\delta'$; and 
\item When $f$ is unate, \hyperref[algo: CheckSamples]{CheckSamples} always returns on any input $a \in f^{-1}(1)$ and any $S \subseteq f^{-1}(1)$.
\end{enumerate}\end{flushleft}
\end{lemma}
\begin{proof}
The query complexity is clear from the description of \hyperref[algo: CheckSamples]{CheckSamples}. Item (2.), the case when $f$ is unate, is also clear given that \hyperref[algo: CheckSamples]{CheckSamples} only rejects when it finds a violation to unateness on line 4.

For part (1.), fix any $f$ and suppose that $S$ contains at least one point with distance more than $2N$ from $a$; then the $z$ picked on line 1 must be such that $|a\Delta z|=\absdist(a, z)>2N$.
Since $N\ge 1$,
  we have that $\{a^{(i)}, z^{(i)}:i\in a\Delta z\}$ are $2\absdist(a, z)$ distinct points.
Given that $N=|f^{-1}(1)|$, the number of 
  $i\in a\Delta z$ such that at least one of $f(a^{(i)})=1$ and $f(z^{(i)})=1$ holds
  is at most $N$.       
    Hence, for $\bi \sim a\Delta z$, with probability at least $(1-{N}/{\absdist(a, z)}) \geq 1/2$ we have $f(a^{(\bi)})=f(z^{(\bi)})=0$.
     So the probability that \hyperref[algo: CheckSamples]{CheckSamples} does not reject is at most $\delta'$.
\end{proof}


\section{Relative-error unateness testing when $N$ is known}

Our main algorithm for the case when $N=f^{-1}(1)$ is known, \hyperref[algo:Known-N]{Unateness-Tester-Known-$N$}, is described in \Cref{algo:Known-N}.

\begin{algorithm}[htb]
\caption{Unateness-Tester-Known-$N$} \label{algo:Known-N}
\vspace{0.15cm}
 \textbf{Input: } $\MQ(f)$ and $\SAMP(f)$ of some function $f$,  $N=|f^{-1}(1)|$ and a distance parameter $\epsilon$. \\
 \begin{tikzpicture}
\draw [thick,dash dot] (0,1) -- (16.5,1);
\end{tikzpicture}
\begin{algorithmic}[1]
\Algphase{Phase 0:}
    \State Draw $\ba \sim \SAMP(f)$.
 \Algphase{Phase 1:}
    \State Let $\bS$ be a set of $\lceil 400 \log N \rceil$ samples  from $\SAMP(f)$. Reject if $\exists x \in \bS$ with $\Delta(x, \ba) > 2 \log N$.
    \State Let $\tilde{\bd} \in \{ 0,1,*\}^n$ be a bias vector where for each $i\in [n]$, $\tilde{\bd}_i$ is defined as:
    $$\tilde{\bd}_i=\begin{cases}
        0 & \text{ if $\Prx_{\bz \in \bS}[\bz_i = 1]<0.25$;} \\
        1 & \text{ if $\Prx_{\bz \in \bS}[\bz_i = 1]>0.75$;} \\
        * & \text{ otherwise.} \\
    \end{cases}$$
 \Algphase{Phase 2:}
    \State Let $\bT$ be a set of $\lceil 30/\eps \rceil$ samples drawn from $\SAMP(f)$. Reject if $\exists x \in \bT$ with $\Delta(x, \ba) > 2 \log N$.
    \State Let $\bM=2 \log N + \Delta(\ba, \tilde{\bd})$.

\Algphase{Phase 3:}
    \State Run \hyperref[algo: BiasedTest]{BiasedTest}($\MQ(f), \SAMP(f)$, $\bM$, $\epsilon/64$, $\tilde{\bd}$). 
    \State Run \hyperref[algo: UnbiasedTest]{UnbiasedTest}($\MQ(f), \SAMP(f)$, $\epsilon/64$, $\tilde{\bd}$).
    \State Halt and accept $f$. 
\end{algorithmic}
\end{algorithm}

The goal of this section is to prove the following theorem, which is a more detailed statement of \Cref{thm:firstpositive}:

\begin{theorem} \label{thm:known-N}
\hyperref[algo:Known-N]{Unateness-Tester-Known-$N$}
  takes as input $\MQ(f)$ and $\SAMP(f)$ of some  function
  $f:\{0,1\}^n\rightarrow \{0,1\}$, the
  value $N=|f^{-1}(1)|$, and a distance parameter $\eps$.
It makes
$\Tilde{O}((\log N)/\epsilon)$ queries and satisfies the following conditions:
When $f$ is unate, it always accepts; 
  when $f$ is $\eps$-far from unate in relative distance, it rejects with probability at least $2/3$.
\end{theorem}

\subsection{High-level overview}

The algorithm \KnownN\ consists of four phases:

\begin{flushleft}\begin{enumerate}
    \item[]\textbf{Phase 0:} The goal is to pick a base point $\ba\in f^{-1}(1)$ for the tester. Recall that by \Cref{lem:small-diameter}, any point in $f^{-1}(1)$ must be close to $\ba$ (Hamming distance at most $2\log N$) if $f$ is unate. The base point $\ba$ will be useful in Phase~1 and Phase~2 since we can now consider only the points from $f^{-1}(1)$ that are close to $\ba$ (i.e., consider the truncated function $f_{2\log N, a}$) and instantly reject if any satisfying assignment of $f$ that is far away from $\ba$ is found.
\end{enumerate}\end{flushleft}

\begin{flushleft}\begin{enumerate}
    \item[]\textbf{Phase 1:} The goal is to compute a bias vector $\tilde{\bd}\in \zos^n$ that is consistent with $d^f$ and is ``close to'' the base point $\ba$.
    Intuitively, $\tilde{\bd}$ is consistent with $d^f$ because the algorithm draws enough number of samples in $\bS$ to correctly estimate the bias of each coordinate with high probability, and the bias vector $\tilde{\bd}$ estimated from $\bS$ is close to $\ba$ since all samples in $\bS$ must be close to $\ba$ (Hamming distance at most $2\log N$).
    The main result we will prove about Phase~1 is the following:
\end{enumerate}\end{flushleft}

\begin{lemma}\label{lem:Phase-1-known-N}
For every $f:\zo^n\to \zo$ and $a\in f^{-1}(1)$, Phase~1 of \hyperref[algo:Known-N]{Unateness-Tester-Known-N} either rejects or returns a vector $\tilde{\bd}\in \{0,1,*\}^n$, and satisfies the following performance guarantees:
\begin{flushleft}\begin{enumerate}
    \item If $f$ is unate, Phase~1 never rejects.

    \item  Phase~1 rejects or both of the following are true:
        $$\left|\unfixed(\tilde{\bd})\right| \leq 8 \log N;$$
        $$\absdist(\ba, \tilde{\bd})=\left|\{i\in \fixed(\tilde{\bd}) \mid \tilde{\bd}_i\neq \ba_i\}\right|\leq 3 \log N.$$
    
    \item With probability at least $0.9$, Phase~1 rejects or $\tilde{\bd}$ is consistent with $d^f$.
  
\end{enumerate}\end{flushleft}
\end{lemma}

\begin{flushleft}\begin{enumerate}
    \item[]\textbf{Phase 2:} The goal of Phase~2 is to set the value of $\bM$ properly to satisfy the pre-condition in \Cref{lem:ftohtof} that $f_{\bM, \tilde{\bd}}$ is close to $f$ in relative distance. The performance guarantee is shown in the following lemma. Intuitively, Phase~1 guarantees that (with high probability) almost every point $x\in f^{-1}(1)$ is close to $\ba$ and $\ba$ is close to $\tilde{\bd}$, so setting $\bM$ as the sum of these two distances is enough to cover almost every point in $f^{-1}(1)$.
\end{enumerate}\end{flushleft}

\begin{restatable}{lemma}{valueofM} \label{lem:Phase-2-known-N}
    For every function $f: \zo^n \to \zo$, Phase~2 either rejects or returns $\bM \in \mathbb{N}$.
    It has the following performance guarantees:
    \begin{enumerate}
        \item Phase~2 never rejects when $f$ is unate.
        \item With probability at least $0.9$, Phase~2 rejects or $\reldist(f, f_{\bM, \tilde{\bd}})\leq  \epsilon/10$.
    \end{enumerate}
\end{restatable}

\begin{flushleft}\begin{enumerate}    
    \item[]\textbf{Phase 3:} This phase simply runs \BiasedTest\ and \UnbiasedTest. 
    Based on the guarantees from Phase~1 and Phase~2, one of the two testers must reject with high probability if $f$ is far from unateness in relative distance; if neither tester rejects, the algorithm accepts $f$.
\end{enumerate}\end{flushleft}

\subsection{Proof of \Cref{thm:known-N} assuming \Cref{lem:Phase-1-known-N} and \Cref{lem:Phase-2-known-N}}

\subsubsection{Query complexity}

We first analyze the query complexity of \hyperref[algo:Known-N]{Unateness-Tester-Known-N}. 

\begin{enumerate}
    \item Phase 0 makes one query to $\SAMP(f)$.
    \item Phase 1 makes $O(\log N)$ queries to $\SAMP(f)$.
    \item Phase 2 makes $O(1/\eps)$ queries to $\SAMP(f)$.
    \item \hyperref[algo: BiasedTest]{BiasedTest} in Phase 3 makes $O(\bM/\epsilon')$ queries by \Cref{theo:mono}, which is $O((\log N)/\epsilon)$, since $\epsilon'=\epsilon/64$ and $\bM = 2 \log N + \Delta(\ba, \tilde{d}) \leq 5 \log N$ (follows from Item~(2.) of \Cref{lem:Phase-1-known-N}).
    \item \hyperref[algo: UnbiasedTest]{UnbiasedTest} in Phase 3 makes $\tilde{O}\left(\frac{|\unfixed(\tilde{\bd})|}{\epsilon'}\right)$ queries by \Cref{thm:correctnessOfUnbiasedTest}, which is $\tilde{O}((\log N)/\epsilon)$ since we have $|\unfixed(\tilde{\bd})| \leq 8\log N$ from Item~(2.) of \Cref{lem:Phase-1-known-N} and $\epsilon'=\epsilon/64$.
\end{enumerate}
The overall query complexity is therefore $\tilde{O}((\log N)/\epsilon)$.

\subsubsection{Non-adaptivity}
\label{subsubsection-nonadaptivity-known-N}

We present \hyperref[algo:Known-N]{Unateness-Tester-Known-N} in a way that is easy to understand, but it may not be directly clear why it is non-adaptive. In this subsection we show how the algorithm can be made non-adaptive without an asymptotic increase in query complexity. The adaptive parts of the algorithm are the following:

\begin{flushleft}\begin{enumerate}

\item In \hyperref[algo: BiasedTest]{BiasedTest} in Phase~3, the algorithm decides whether to run \hyperref[algo: ConfirmDirection]{ConfirmDirection} based on the value of $f(\bz^{(i)})$ (Step 8 and 9). This step can be made non-adaptive by always running \hyperref[algo: ConfirmDirection]{ConfirmDirection} regardless of $f(z^{(i)})$, which only blows up the number of queries by a constant (recall that \hyperref[algo: ConfirmDirection]{ConfirmDirection} only makes constantly many queries).

\item In \BiasedTest\ and \UnbiasedTest\ in Phase~3, the query complexity depends on $\tilde{\bd}$ and $\bM$ from Phase~1 and Phase~2. Although Phases~0,1 and 2 only use the sampling oracle $\SAMP(f)$, the algorithm still adaptively decides the number of future samples 
in Phase~3 based on the results of samples from previous phases. To address the adaptivity here, the algorithm can always draw the maximum number of samples required in Phase~3, i.e., consider $|\unfixed(\tilde{\bd})|$ as $8\log N$ and $\bM$ as $5\log N$ where the upper bounds are guaranteed by Item~(2.) of \Cref{lem:Phase-1-known-N}.

\end{enumerate}\end{flushleft}

\subsubsection{Completeness}

By Item~(1.) of \Cref{lem:Phase-1-known-N} and Item~(1.) of \Cref{lem:Phase-2-known-N}, Phases~0,1 and 2 never reject a unate function. By \Cref{theo:mono} and \Cref{thm:correctnessOfUnbiasedTest}, \hyperref[algo: BiasedTest]{BiasedTest} and \hyperref[algo: UnbiasedTest]{UnbiasedTest} in Phase~3 reject only when they find two edges that violate the unateness along some direction, which means they also never reject a unate function. Thererfore, \hyperref[algo:Known-N]{Unateness-Tester-Known-N} always accepts when the input function $f$ is a unate function.
   
\subsubsection{Soundness}

Assume that the input function $f$ is $\eps$-far from unate in relative distance. By Item~(3.) of \Cref{lem:Phase-1-known-N}, with probability at least $0.9$, either $f$ is already rejected in Phase~1, or we reach Phase~2 with $\tilde{\bd}$ that is consistent with $d^f$ (which is one of the pre-conditions of \Cref{theo:mono}). 
Now suppose we have not rejected before Phase~2.
Then by \Cref{lem:Phase-2-known-N}, with probability at least $0.9$, either $f$ is rejected in Phase~2, or we reach Phase~3 with $\bM$ and $\tilde{\bd}$ that satisfies $\reldist(f, f_{\bM,\tilde{\bd}})\le \eps/10$.
Now, condition on reaching Phase~3 with $\tilde{\bd}$ that is consistent with $d^f$ and $h=f_{\bM,\tilde{\bd}}$ that satisfies $\reldist(f, h)\leq \eps/10$.
We have by \Cref{lem:ftohtof} that either \Cref{eq:many-viols-on-unfixed-idxs_f} or \Cref{eq:many-viols-on-fixed-idxs_f_Md} holds. 
It then follows from \Cref{theo:mono} and \Cref{thm:correctnessOfUnbiasedTest} that either \hyperref[algo: BiasedTest]{BiasedTest} or \hyperref[algo: UnbiasedTest]{UnbiasedTest} rejects with probability at least~$0.99$.
By a union bound, the algorithm rejects with overall probability at least $0.79$.

\subsection{Proofs of \Cref{lem:Phase-1-known-N} and \Cref{lem:Phase-2-known-N}}

In this section, we prove the performance guarantees of Phase~1 and Phase~2 as summarized in \Cref{lem:Phase-1-known-N} and \Cref{lem:Phase-2-known-N}. We first prove \Cref{lem:Phase-1-known-N} through the following series of claims.

\begin{claim}[Item~(1.) of \Cref{lem:Phase-1-known-N}] \label{claim:Phase-1-1-known-N}
    If $f$ is unate, Phase~1 never rejects.
\end{claim}

\begin{proof}
    Phase~1 only rejects when it finds $\absdist(x,a)>2\log N$ where $a, x\in f^{-1}(1)$. By \Cref{lem:small-diameter}, this cannot happen when $f$ is unate.
\end{proof}

\begin{claim}[Item~(2.) of \Cref{lem:Phase-1-known-N}] \label{claim:Phase-1-3-known-N}
    If Phase~1 does not reject (i.e.~returns $\tilde{\bd}$), then
        $$\left|\unfixed(\tilde{\bd})\right| \leq 8 \log N;$$
        $$\absdist(\ba, \tilde{\bd})=\left|\{i\in \fixed(\tilde{\bd}) \mid \tilde{\bd}_i\neq \ba_i\}\right|\leq 3 \log N.$$
\end{claim}

\begin{proof}
    Fix an arbitrary choice of $\bS\subseteq f^{-1}(1)$ such that $\Delta(\ba, x)\leq 2\log N$ for each $x\in \bS$. 
    We clearly have $\Ex_{\bx\sim \bS}[\absdist(\ba, \bx)]\leq 2\log N$. On the other hand, we have
    $$\Ex_{\bx\sim \bS}[\absdist(\bx, \ba)]=\sum_{i\in [n]}\Prx_{\bx\sim \bS}[\bx_i\neq \ba_i]\geq \sum_{i\in \unfixed(\tilde{\bd})}\Prx_{\bx\sim \bS}[\bx_i\neq \ba_i]\geq 0.25 \cdot \left|\unfixed(\tilde{\bd})\right|.$$
    Therefore, we must have $\left|\unfixed(\tilde{\bd})\right|\leq 8\log N$. Similarly, for $\absdist(\ba, \tilde{\bd})$, we have $$\Ex_{\bx\sim \bS}[\absdist(\bx, \ba)]\geq \sum_{i\in \fixed(\tilde{\bd}): \tilde{\bd}_i\neq \ba_i}\Prx_{\bx\sim \bS}[\bx_i\neq \ba_i]\geq 0.75 \absdist(\ba, \tilde{\bd}),$$
    and thus $\absdist(\ba, \tilde{\bd})\leq 3\log N$.
\end{proof}

\begin{claim}[Item~(3.) of \Cref{lem:Phase-1-known-N}] \label{claim:Phase-1-4-known-N}
    With probability at least $0.9$, Phase~1 either rejects or returns $\tilde{\bd}$ that is consistent with $d^f$.
\end{claim}

\begin{proof}
    We first consider the case $N\geq n$ where the number of samples satisfies $|\bS|\geq 400\log n$. If $\tilde{\bd}$ is not consistent with $d^f$, then there exists some $i\in [n]$ such that $\tilde{\bd}_i\neq *$ but $d^f_i\in \{*, 1-\tilde{\bd}_i\}$, and thus we must have
    \begin{equation}
        \left|\Prx_{\bz \sim \bS}[\bz_i = 1]-\Prx_{\bz \sim f^{-1}(1)}[\bz_i = 1]\right|>0.15 \label{eq:useme}
    \end{equation}
    where $0.15$ is the difference between the thresholds for setting $*$'s in $\tilde{\bd}$ (see line~3 of Phase~1) and $d^f$ (see \Cref{def:dirVector}). 
    By a standard Chernoff bound, for each coordinate $i\in [n]$, \Cref{eq:useme} happens with probability at most $2e^{-2|S|(0.15)^2}\leq 0.1n^{-1}$. 
    Then, by a union bound over all $i \in [n]$, the probability that $\tilde{\bd}$ is not consistent with $d^f$ is at most $0.1$.

    The second case is that $N < n$ and $\rd(f, f_{2\log N, a}) > 0.05$.
    In this case, each random sample $\bx\sim \SAMP(f)$ satisfies $\absdist(a, \bx)>2\log N$ with probability at least $0.05$, and thus the probability that Phase~1 rejects is at least $1-(0.95)^{|\bS|}\geq 1-(0.95)^{400}\geq 0.9.$

    The third and final case is that $N < n$ and $\rd(f, f_{2 \log N, a}) < 0.05$.
    Assume we do not reject.
    Then, $\bS\subseteq f^{-1}_{2\log N, a}(1)$.
    Consider the set of non-trivial coordinates (\Cref{defn:nontrivial-coordinates}) for which $\Prx_{\bz \sim f^{-1}_{2\log N, a}(1)}[\bz_i = \ba_i] \neq 1$. 
    We have by \Cref{claim:trunc-fns-nontrivial-coord-bound} that there are at most $2N\log N$ non-trivial coordinates. 
    If $\tilde{\bd}$ is not consistent with $d^f$, then by \Cref{lem:loss-from-truncation} and the triangle inequality applied to \Cref{eq:useme}, there exists some $i\in[n]$ such that $$\left|\Prx_{\bz \sim \bS}[\bz_i = 1]-\Prx_{\bz \sim f^{-1}_{2\log N, a}(1)}[\bz_i = 1]\right|>0.1.$$ 
    For each trivial coordinate $i$ where $\Prx_{\bz \sim f^{-1}_{2\log N, a}(1)}[\bz_i = \ba_i]=1$, the above event can never happen since the two probabilities are always the same; for each non-trivial coordinate, it happens with probaility at most $2e^{-2|S|(0.1)^2}\leq 0.01N^{-2}$ by a standard Chernoff bound. 
    By a union bound over the non-trivial coordinates, the probability that $\tilde{\bd}$ is not consistent with $d^f$ is at most $0.01N^{-2}\cdot 2N\log N\leq 0.1$.
\end{proof}

Putting everything together, \Cref{lem:Phase-1-known-N} follows directly. Now we prove \Cref{lem:Phase-2-known-N}, whose statement we first recall:

\valueofM*

\begin{proof}
    Phase~2 only rejects when it finds $\absdist(a, x)>2\log N$ where $a, x\in f^{-1}(1)$. By \Cref{lem:small-diameter}, this cannot happen when $f$ is unate, so we have item (1).

    For item (2), suppose first that $\rd(f, f_{2 \log N, \ba}) > \eps/10$. 
    Then, each random sample $\bx \sim \SAMP(f)$ satisfies $\Delta(\ba, \bx) > 2 \log N$ with probability at least $\eps/10$, and hence the probability that Phase~2 rejects is at least $1 - (1-\eps/10)^{|\bT|} \geq 1 - \exp(\frac{-\eps |\bT|}{10}) \geq 1 - \exp(-3) \geq 0.9$. So suppose instead that $\rd(f, f_{2 \log N, \ba}) \leq \eps/10$.
    Recalling that $\bM := 2 \log N + \Delta(\ba, \tilde{\bd})$, since  $f^{-1}_{2 \log N, \ba}(1) \subseteq f^{-1}_{\bM, \tilde{\bd}}(1) \subseteq f^{-1}(1)$, we have that $\rd(f, f_{\bM, \tilde{\bd}})\leq \rd(f, f_{2 \log N, \ba}) \leq \eps/10$.
\end{proof}

\newcommand{\CheckSample}{\hyperref[algo: CheckSamples]{CheckSamples}}
\newcommand{\Preprocessing}
{\hyperref[alg: Preprocessing]{Preprocessing}}

\section{Unateness testing in relative distance when $N$ is unknown} \label{sec:unknown-N}

Our main algorithm, \hyperref[algo:unknown]{Unateness-Tester-Unknown-$N$}, for the case when $N=|f^{-1}(1)|$ is unknown~to the algorithm,
  is described in \Cref{algo:unknown}.
The goal of this section is to prove the following theorem, which is a more detailed statement of \Cref{thm:secondpositive}:

\begin{theorem} \label{thm:unknown-N}
    \hyperref[algo:unknown]{Unateness-Tester-Unknown-$N$} takes as input $\MQ(f)$ and $\SAMP(f)$ of some function $f:\{0,1\}^n\rightarrow \{0,1\}$, a distance parameter $\eps$ and an error parameter $\delta$.
    It makes no more than
    $$
\tilde{O}\left(\frac{\log N}{\eps}+ \log(1/\delta)\right)
    $$ queries with probability at least $1-\delta$, and it  satisfies the following conditions:
    It always accepts when $f$ is unate;
    when $f$ is $\eps$-far from unate in relative distance, it rejects with probability at least $2/3$.
\end{theorem}

As can be seen in \Cref{algo:unknown}, we divide \hyperref[algo:unknown]{Unateness-Tester-Unknown-$N$}
  into four phases.
As it will become clear soon in \Cref{sec:unknown-overview} (where we give a more detailed overview
  of the four phases), 
  performance guarantees achieved by Phases 1 and 2 of
  \hyperref[algo:unknown]{Unateness-Tester-Unknown-$N$} are similar to those of \hyperref[algo:Known-N]{Unateness-Tester-Known-$N$}:
(1) Phase 1 aims to find a bias vector $\tilde{\bd}$ that is consistent with $d^f$ and 
  has very few unfixed coordinates; 
(2) Phase 2 aims to return an integer $\bM$ such that the truncated function
  $f_{\bM,\tilde{\bd}}$ is close to $f$ and at the same time $\bM$ is bounded by~$O(\log N)$. 
The new challenge for both of these tasks is that we need to make sure that the number of 
  queries made by the algorithm is logarithmic in $N$ with high probability, even though $N$ is unknown.
To this end, (1) \hyperref[algo:unknown]{Unateness-Tester-Unknown-$N$} has a more substantial 
  Phase 0 that performs a preliminary check on $f$ using a new subroutine called 
  \hyperref[algo: CheckSamples]{CheckSamples}; and (2) Phases 1 and 2 are implemented by
  two more complex subroutines called \hyperref[algo: IterativeBias]{IterativeBias} and 
  \hyperref[alg: Preprocessing]{Preprocessing} with more delicate performance guarantees
  (see \Cref{sec:unknown-overview}).
  
Once Phases 0, 1  and 2 are in place,
the soundness of \hyperref[algo:unknown]{Unateness-Tester-Unknown-$N$} 
follows again from the fact that it rejects only when a violation to unateness is found.
The completeness follows by using arguments similar to those in the analysis of \hyperref[algo:Known-N]{Unateness-Tester-Known-$N$}: If a function $f$ that is far from unate in relative distances passes both Phase 1 and 2
  with $\tilde{\bd}$ and $\bM$ that satisfy conditions sketched above,
  then either \hyperref[algo: BiasedTest]{BiasedTest} or 
  \hyperref[algo: UnbiasedTest]{UnbiasedTest} in Phase 3 rejects with high probability.

\begin{algorithm}[t!]
\caption{Unateness-Tester-Unknown-$N$}\label{algo:unknown}
 \textbf{Input: }$\MQ(f)$ and $\SAMP(f)$ of some function $f$, a distance parameter $\epsilon$ and an error parameter $\delta$. \\
 \textbf{Output: } Either accept or reject.\\
 \begin{tikzpicture}
\draw [thick,dash dot] (0,1) -- (16.5,1);
\end{tikzpicture}
\begin{algorithmic}[1]
\Algphase{Phase 0:}
\State Draw $\ba \sim \SAMP(f)$.
\State Draw a set $\bS$ of $100 \lceil \log(1/\delta) \rceil$ points from $\SAMP(f)$.
\State Run \hyperref[algo: CheckSamples]{CheckSamples}($\MQ(f),\ba, \bS, \delta/12)$; halt and reject if it rejects; continue if it accepts. 
 \Algphase{Phase 1:}
\State Let $\tilde{\bd}=\text{\hyperref[algo: IterativeBias]{IterativeBias}}(\MQ(f), \SAMP(f), \ba, \delta)$. 
 \Algphase{Phase 2:}

\State Let $\bM=\text{\hyperref[alg: Preprocessing]{Preprocessing}}(\MQ(f), \SAMP(f), \tilde{\bd}, \epsilon, \delta)$. 
\Algphase{Phase 3:}
\State Run $\text{\hyperref[algo: BiasedTest]{BiasedTest}}(\MQ(f), \SAMP(f) ,  \bM ,  \epsilon/64 ,  \tilde{\bd} )$. 
\State Run $\text{\hyperref[algo: UnbiasedTest]{UnbiasedTest}}(\MQ(f ), \SAMP(f), \epsilon/64,  \tilde{\bd})$.
\State Accept $f$. 
\end{algorithmic}
\end{algorithm}

\subsection{High-level overview}\label{sec:unknown-overview}

The algorithm \hyperref[algo:unknown]{Unateness-Tester-Unknown-$N$}
consists of four phases:
\begin{flushleft}\begin{enumerate}
    \item[]\textbf{Phase 0:} Phase 0 starts by drawing a point $\ba \sim f^{-1}(1)$. 
    It then draws a set $\bS$ of random points from $f^{-1}(1)$ and runs a subroutine \hyperref[algo: CheckSamples]{CheckSamples} (see \Cref{sec:analysis-CheckSamples} for its description and analysis).
    The goal of Phase 0 is to reject with high probability when the truncated function $f_{2N,\ba}$ is far from $f$.
    (Note that this can happen only when 
    $f$ is not unate, as by \Cref{lem:small-diameter} and using $2\log N\le 2N$, we have that $f$ is identical to $f_{2N,a}$ for every $a\in f^{-1}(1)$ when $f$ is unate.)
    Thinking about the function $f_{2N,\ba}$ and in particular, being able to assume that $f_{2N,\ba}$ is close to $f$ will be useful during Phase~1. We summarize performance guarantees of Phase~0 in the following lemma:
\end{enumerate}\end{flushleft}
\begin{restatable}{lemma}{phaseO}\label{lem:phase0}
 Phase~0 of \hyperref[algo:unknown]{Unateness-Tester-Unknown-$N$} always makes $O(\log(1/\delta))$ queries and satisfies the following performance guarantees:
\begin{flushleft}\begin{enumerate}
\item If $f$ is unate, Phase~0 never rejects; and
\item For any function $f$, the probability of \hyperref[algo:unknown]{Unateness-Tester-Unknown-$N$} reaching Phase~1 with a string $\ba$ such that $\reldist(f,f_{2N,\ba}) \geq 0.05$ is at most $\delta/6$.      
\end{enumerate}\end{flushleft}
\end{restatable}

\begin{flushleft}\begin{enumerate}
    \item[] \textbf{Phase 1:} The goal of Phase 1 is to compute a vector $\tilde{\bd}\in \{0,1,*\}^n$ that satisfies (1) $\tilde{\bd}$ does not have too many (i.e., $O(\log N)$) unfixed coordinates and (2) $\tilde{\bd}$ is consistent with the bias vector $d^f$ of $f$.
This is achieved by the new subroutine called 
  \IterativeBias.
We summarize performance guarantees of \IterativeBias\  in the following main technical lemma. 
 The design and analysis of \IterativeBias\ is the most challenging part of the paper, where we need to make sure that the number of queries used is at most $\tilde{O}(\log(N/\delta))$ with high probability, even though $N$ is not given to the algorithm. 
\end{enumerate}\end{flushleft}

\begin{lemma}\label{lem:maintechnical}
Given any function $f$, $a\in f^{-1}(1)$ and $\delta>0$,
$\text{\IterativeBias}(\MQ(f),\SAMP(f), a, \delta)$ either accepts, rejects or returns a   $\tilde{\bd}\in \{0,1,*\}^n$, and satisfies the following performance guarantees:
\begin{flushleft}\begin{enumerate}
    \item If $f$ is unate, \IterativeBias\ never rejects;
    \item  For any function $f$, \IterativeBias\ accepts  with probability at most $0.1$;
    \item For any function $f$ and $a \in f^{-1}(1)$
    such that $\rd(f,f_{2N,a}) \leq 0.05$,
    the probability of \IterativeBias\  returning a $\tilde{\bd}\in \{0,1,*\}^n$ that is not consistent with $d^f$ is at most $\delta/6$;

    \item For any function $f$ and $a \in f^{-1}(1)$ such that  $\rd(f,f_{2N,a}) \leq 0.05$, the probability of \IterativeBias\ 
    returning a $\tilde{\bd}\in \{0,1,*\}^n$ with
    \begin{equation}\label{eq:smallstar}
    \big|\unfixed(\tilde{\bd}) \big| > 100 \log N
    \end{equation}
    is at most $\delta/6$.
   \item Query complexity: For any function $f$,  with probability at least $1-\delta/6$, \IterativeBias\ makes at most $\tilde{O}(\log (N/\delta))$ many queries to $\MQ(f)$ and $\SAMP(f)$.
  
\end{enumerate}\end{flushleft}
\end{lemma}  

\begin{flushleft}\begin{enumerate}

\item[] \textbf{Phase 2:} Assume that Phase~1 does not reject and returns a vector $\tilde{\bd}\in \{0,1,*\}^n$ that is consistent with $f$ and has few unfixed coordinates.  
Phase 2 runs a new subroutine, \hyperref[alg: Preprocessing]{Preprocessing}, to obtain a nonnegative integer $\bM$ such that (1) $\bM=O(\log N)$ (but $\bM$ could be asymptotically smaller than $\log N$) and (2) the $(\bM,\tilde{\bd})$-truncated function $\smash{h=f_{\bM,\tilde{\bd}}}$ obtained from $f$  is close to $f$.
We summarize performance guarantees of \hyperref[alg: Preprocessing]{Preprocessing} in the lemma below:
\end{enumerate}\end{flushleft}
\begin{lemma}\label{lem:preprocessing}
Given any function $f$,  $\tilde{d}\in \{0,1,*\}^n$ and parameters $\eps$ and $\delta$, \hyperref[alg: Preprocessing]{Preprocessing}$(\MQ(f),$ $\SAMP(f),\tilde{d},\eps,\delta)$  either rejects or returns a nonnegative integer $\bM$, and satisfies the following~performance guarantees:
\begin{flushleft}\begin{enumerate}
\item If $f$ is unate \hyperref[alg: Preprocessing]{Preprocessing} never rejects. 
\item For any function $f$, the probability of \Preprocessing\ returning an $\bM$ with $\reldist(f,\smash{f_{\bM,\tilde{d}}})>\eps/10$ is at most $0.1$. 
\item For any function $f$ and $\tilde{d}$ such that  $\tilde{d}$ is consistent with $d^f$, the probability of \hyperref[alg: Preprocessing]{Preprocessing}  returning an  $\bM> 2\log N$ is at most $\delta/6$. 
\item  For any function $f$ and 
  $\tilde{d}$ such that $\tilde{d}$ is consistent with $d^f$, with probability at least $1-\delta/6$,  \hyperref[alg: Preprocessing]{Preprocessing}\ makes at most  $$O\left((1/\eps)+ \log(N)+\log(1/\delta)\right)$$ many queries to $\MQ(f)$ and $\SAMP(f)$. 
\end{enumerate}\end{flushleft}
\end{lemma}

\begin{flushleft}\begin{enumerate}
\item[] \textbf{Phase 3:} This phase simply runs \BiasedTest\ and \UnbiasedTest. 
Based on the performance guarantees from Phase~1 and Phase~2 discussed above and an analysis similar to that of Phase 3 of \hyperref[algo:Known-N]{Unateness-Tester-Known-$N$}, one of these two testers must reject with high probability when $f$ is far from unateness in relative distance.

\end{enumerate}\end{flushleft}

\subsection{Proof of \Cref{thm:unknown-N} assuming  \Cref{lem:phase0}, \Cref{lem:maintechnical} and 
\Cref{lem:preprocessing}}

We prove \Cref{lem:phase0} in \Cref{sec:Phase0} (which is based on the analysis of \CheckSample\ in 
\Cref{sec:analysis-CheckSamples}).
We then prove \Cref{lem:maintechnical} in \Cref{sec:maintechnical} and \Cref{lem:preprocessing} in \Cref{sec:preprocessing}, respectively.
We now use those lemmas to prove \Cref{thm:unknown-N}. 

\subsubsection{Query complexity}
We start by analyzing the query complexity of \hyperref[algo:unknown]{Unateness-Tester-Unknown-$N$}. Given any function $f$, consider the following bad events of \hyperref[algo:unknown]{Unateness-Tester-Unknown-$N$} running on $f$:
\begin{flushleft}\begin{enumerate}
\item The algorithm reaches Phase~1 with $\ba$ such that $\rd(f,f_{2N,\ba}) > 0.05$, which by \Cref{lem:phase0} happens with probability at most $\delta/6$.  
\item Phase 1 uses more than $\tilde{O}(\log(N/\delta))$ queries, which by \Cref{lem:maintechnical} happens with probability at most $\delta/6$;
\item 
Assuming that the $\ba$ from Phase~0 satisfies $\rd(f,f_{2N,\ba}) \leq 0.05$, 
Phase 1 returns a $\tilde{\bd}$ with more than $100\log N$ many $*$'s, which by \Cref{lem:maintechnical} happens with probability at most $\delta/6$; 
\item 
Assuming that the $\ba$ from Phase~0 satisfies $\rd(f,f_{2N,\ba}) \leq 0.05$, 
Phase 1 returns a $\tilde{\bd}$ that is not consistent with $d^f$, which by \Cref{lem:maintechnical} happens with probability at most $\delta/6$; 
\item Assuming the $\tilde{\bd}$ from Phase~1 is consistent with $d^f$, Phase 2 returns an $\bM> 2\log N$, which by \Cref{lem:preprocessing} happens with probability at most $\delta/6$;  and
\item Assuming the $\tilde{\bd}$ from Phase~1 is consistent with $d^f$, Phase 2 uses more than $O(1/\epsilon+\log(N)$ $+\log(1/\delta))$ queries, which by \Cref{lem:preprocessing} happens with probability at most $\delta/6$.
\end{enumerate}\end{flushleft}
By a union bound, with probability at least $1 - \delta$ none of these 6 bad events occurs. When this is the case, by \Cref{lem:maintechnical,lem:preprocessing} the query complexity of Phases 0,1 and 2 together is at most
$$
O\bigg(\log(1/\delta)\bigg) +\tilde{O}\left(\log\left(\frac{N}{\delta}\right)\right) + O\left(\frac{1}{\eps}+\log(N)+\log\left(\frac{1}{\delta}\right)\right),
$$
and the query complexity of Phase 3 (by \Cref{theo:mono} and \Cref{thm:correctnessOfUnbiasedTest}) is at most
$$
O\left(\frac{\bM}{\eps}\right)+\tilde{O}\left(\frac{|\unfixed(\tilde{\bd})|}{\eps}\right)= \tilde{O}\left(\frac{\log N}{\eps}\right),
$$
from which the query complexity upper bound stated in \Cref{thm:unknown-N} follows.

\subsubsection{Completeness}

Assume that the function $f$ is unate. By \Cref{lem:phase0}, \Cref{lem:maintechnical} and \Cref{lem:preprocessing}, none of Phase~0, 1 and 2 can reject $f$. By \Cref{theo:mono} and \Cref{thm:correctnessOfUnbiasedTest}, Phase~3 can only reject if $f$ is not unate.

\subsubsection{Soundness}

\noindent Assume that the function $f$ is $\eps$-far from unate in relative distance. 
By \Cref{lem:maintechnical} (the second item), we have that the probability of the algorithm reaching Phase 1 and accepting $f$ during Phase 1 is at most $0.1$.
Assuming that this does not happen, then by \Cref{lem:preprocessing} (the second item), with probability at least $0.9$, either $f$ is already rejected in Phase 0, 1 or 2, or we reach Phase 3 with $\bM$ and $\tilde{\bd}$ that satisfy $\smash{\reldist(f, f_{\bM,\tilde{\bd}})\le \eps/10}$.
Assuming that Phase 3 is reached with $\bM$ and $\tilde{\bd}$ that satisfy $\smash{\reldist(f,f_{\bM,\tilde{\bd}})\le \eps/10}$,  we have by \Cref{lem:ftohtof} that either \Cref{eq:many-viols-on-unfixed-idxs_f} or \Cref{eq:many-viols-on-fixed-idxs_f_Md} holds.
 It then follows directly from \Cref{theo:mono} and \Cref{thm:correctnessOfUnbiasedTest} that either BiasedTest or UnbiasedTest rejects with probability at least~$0.99$.
So by a union bound the algorithm rejects with probability at least $0.79$.

\subsection{Phase~0: Proof of \Cref{lem:phase0}}\label{sec:Phase0}
    Line $2$ of Phase~0 uses $100 \lceil \log(1/\delta) \rceil$ samples while the call to \hyperref[algo: CheckSamples]{CheckSamples} on line~3 uses $O(\log(1/\delta))$ samples by \Cref{lem:CheckSamplesNoRej}. By the same lemma, \hyperref[algo: CheckSamples]{CheckSamples} always accepts when $a \in f^{-1}(1)$ and $f$ is unate. So Phase~0 never rejects when $f$ is unate. 
    
    Finally, assume that for the drawn $\ba$ we have $\reldist(f,f_{2N,\ba}) \geq 0.05$. Since $f_{2N,\ba}^{-1} \subseteq f^{-1}(1)$, we have $|f^{-1}(1) \setminus f_{2N,\ba}^{-1}| \geq 0.05N$. Hence, the probability that are no points $\bz \in \bS$ with $\absdist(\ba, \bz) > 2N$ is at most $(1-0.05)^{100 \lceil \log(1/\delta) \rceil} \leq\delta/12$.

    Assuming that there exists some $\bz^* \in \bS$ with $\absdist(\ba, \bz^*) >2N$, by  \Cref{lem:CheckSamplesNoRej}, \hyperref[algo: CheckSamples]{CheckSamples} will reject with probability at least $1-\delta/12$. Hence, by a union bound, the probability the algorithm reaches Phase~1 but $\reldist(f,f_{2N,\ba}) \geq 0.05$ is at most $\delta/6$.
This concludes the proof of \Cref{lem:phase0}. \qed

\subsection{Phase~1 and \hyperref[algo: IterativeBias]{IterativeBias}: Proof of \Cref{lem:maintechnical}}\label{sec:maintechnical}

Before entering into the proof we give some intution.
Recall that in the known-$N$ case, the algorithm relied on being given the value on $N$; in particular, it used a set $\bS$ of $O(\log N)$ samples to estimate probabilities with high enough confidence so as to allow a union bound over up to $N$ coordinates in the proof of \Cref{claim:Phase-1-4-known-N}.

In the current context, where we are not given the value of $N$, we take a more refined approach.  This is essentially to iteratively draw a larger and larger set of samples, checking as we go whether the set is large enough.  We ultimately need to handle a union bound over at most $2N^2 = \poly(N)$ coordinates, which means that $O(\log N)$ samples is enough. 

To make this work, we will use the function $f_{2N,a}$. While we don't know $N$, we know from 
\Cref{claim:trunc-fns-nontrivial-coord-bound}
that this function depends on at most $2N^2$ variables and furthermore, the \hyperref[algo: CheckSamples]{CheckSamples} subroutine roughly allows us to have access to uniformly random samples from $f_{2N,a}^{-1}(1)$ (recall \Cref{lem:CheckSamplesNoRej}).

\subsubsection{Phase 1 of \IterativeBias}
Our goal in this subsection is to {establish results that will be used to} prove Items (2.) and (3.) of \Cref{lem:maintechnical} in \Cref{sec:puttogether}. We first prove a series of lemmas assuming we're actually able to sample from $f^{-1}_{2N,a}(1)$ uniformly at random (we will get rid of this assumption in \Cref{sec:puttogether}).

\begin{lemma}\label{lem:IterativeBiasWithAssumption2}
    Let $f$ be an arbitrary function. Assume that during each execution of Phase 1.1 of the initial for-loop of \IterativeBias, the two sets $\bS, \bT$ are drawn by sampling uniformly random points from $f_{2N,a}^{-1}(1)$ (instead of $f^{-1}(1)$). 
    Then, the probability of Phase 1 accepting on line 8 is at most $1/20$.
\end{lemma}
\begin{proof}
    Assume that $\bS, \bT$ are drawn by sampling uniformly random points in $f_{2N,a}^{-1}(1)$. When $\ell=10\log(n)$, it follows by a standard Chernoff bound, that for any $i \in [n]$ we have:

    $$\abs{\bp_i - \Prx_{\bz \sim f_{2N,a}^{-1}(1)}{[\bz_i=1]}} \leq 1/40\quad \text{and}\quad \abs{\bp_i' - \Prx_{\bz \sim f_{2N,a}^{-1}(1)}{[\bz_i=1]}} \leq 1/40,$$
    with probability at least $1-20/n$, in which case we have $|\bp_i -\bp_i'| \leq 1/20$. 
    By a union bound over all $[n]$ coordinates, we thus have that when $\ell=10\log(n)$, the algorithm would pass the check of line~5 (and thus accept) with probability at most $1/20$.  
\end{proof}

\begin{lemma}\label{lem:IterativeBiasWithAssumption}
Let $f$ be an arbitrary function. As in \Cref{lem:IterativeBiasWithAssumption2}, assume that during each execution of Phase 1.1 of the initial for-loop of \IterativeBias, the two sets $\bS, \bT$ are drawn by sampling uniformly random points from $f_{2N,a}^{-1}(1)$. Then, the probability of \IterativeBias\ reaching Phase 2 with at least one coordinate $i \in [n]$ satisfying
\begin{equation}\label{eq:close_p}
\left|\bp_i - \Prx_{\bz \sim f_{2N,a}^{-1}(1)}[\bz_i=1]\right| > 0.1     
\end{equation}
is at most $\delta/12$. 
\end{lemma}
\begin{proof} 
   A necessary condition for \IterativeBias\ to reach Phase 1.2 with $i \in [n]$ satisfying \Cref{eq:close_p} is that there exist some $\ell$ such that during the $\ell$-iteration of the initial for-loop, the following holds: Letting $i^*$ be the first index in $[n]$ such that at the end of Phase~1.1 we have $$\left|\bp_{i^*} - \Prx_{\bz \sim f_{2N,a}^{-1}(1)}[\bz_i=1]\right| \leq 1/20,$$ 
   we must have that in Phase~1.2   $$\big|{\bp_{i^*}}- \bp_{i^*}'\big|\leq 1/20.\footnote{Observe that by lines~5-6, if we we had $\smash|{\bp_{i^*}}- \bp_{i^*}'|>1/20$, then we would not enter Phase 2.}
   $$
In particular, we must have that $|\bp_{i^*}' - \Prx_{\bz \sim f_{2N,a}^{-1}(1)}[\bz_i=1]| \geq 0.05$. But by a standard Chernoff bound (and setting the constant hidden in $k=O(\ell + \log(1/\delta)$ sufficiently large), the probability of the event above happening during the loop for $\ell$ is at most $2^{- \ell -\log(24/\delta)}$.

So by a union bound over all loops, the probability of \IterativeBias\ returning reaching Phase~2 with some $i \in [n]$ violating \Cref{eq:close_p} is at most:
$$
2^{-\log(24/\delta)} \cdot \sum_{\ell\ge 0} 2^{- \ell }\le 2\cdot 2^{-\log(24/\delta)} \leq \delta/12.
$$ 
This finishes the proof of the lemma.
\end{proof}

\begin{claim}\label{clm:ananas}
Assume that $\rd(f,f_{2N,a}) \leq 0.05$ and that \IterativeBias\ reaches Phase~2 with
\begin{equation}\abs{ \bp_i - \Prx_{\bz \sim f_{2N,a}^{-1}(1)}[\bz_i=1]} \leq 0.1,\quad\text{for all $i\in [n]$.}
\label{eq:pickle}
\end{equation}
Then we have that:
$$\abs{ \bp_i - \Prx_{\bz \sim f^{-1}(1)}[\bz_i=1]} \leq 0.15
\quad{\text{for all $i\in [n]$.}}$$
\end{claim}
\begin{proof}
    This follows by \Cref{lem:close_trunc_biases_are_similar_to_fn} and the triangle inequality.
\end{proof}

We have the following simple corollary.
\begin{corollary}\label{cor:consistent}
Assume that $\rd(f,f_{2N,a}) \leq 0.05$ and that \IterativeBias\ reaches Phase~2 with \Cref{eq:pickle} holding.
Then either the algorithm rejects or it returns a vector $\tilde{\bd}$ that is consistent with $d^f$. 
\end{corollary}
\begin{proof}
   By \Cref{clm:ananas} we have $\big| \bp_i - \Prx_{\bz \sim f^{-1}(1)}[\bz_i=1]\big| \leq 0.15$ for all $i \in [n]$. Hence, since 0.75 - 0.15 = 0.6 (recalling \Cref{def:dirVector}),
   it must be that $\tilde{\bd}$ is consistent with $d^f$. So in {Phase~2}, if we do not reject, then we return a $\tilde{\bd}$ that is consistent with $d^f$. 
\end{proof}

\subsubsection{{Phase 2}: Bounding the number of $*$ in $\tilde{d}$}
We need the following definition: 
\begin{definition}
 Given $f: \zo^n \to \zo$, for each $i \in [n]$ we let $E_i$ denote the set of all pairs $(x,x^{(i)})$ with $f(x)=1$. And we let $F_i \subseteq E_i$ be the set of all pairs $(x,x^{(i)})$ with $f(x)=f(x^{(i)})=1$. 
\end{definition}

We observe that for each $i \in [n]$, we can sample a pair $(\bx,\by)\in E_i$ uniformly at random by sampling a uniform $\bz \sim f^{-1}(1)$ and returning $(\bz,\bz^{(i)})$.  

We will use the following well-known result:  
\begin{lemma}
[Edge-isoperimetric inequality of Harper \cite{Harper64},
Bernstein \cite{Bernstein67},
Lindsey~\cite{Lin64}, 
and Hart \cite{Hart76}]
\label{lem:points-edges-hypercube}
    Given a set $S \subseteq \zo^n$, there are at most
    $|S|\log|S|$ many pairs $(x,y) \in S \times S$ such that $x$ and $y$ are neighbors in the hypercube.
\end{lemma}

\Cref{lem:points-edges-hypercube} easily yields the following:

\begin{corollary}\label{lem:lot_of_bad_in_B}
    Suppose that $I \subseteq [n]$ satisfies  $|I|>100\log N$, where as usual $N=|f^{-1}(1)|$. When $\bi \sim I$ uniformly at random, we have  $|F_{\bi}| \leq N/20$ with probability at least $1/2$.
\end{corollary}
\begin{proof}
    Suppose that $|I|>100\log N$.  Let $I^* \subseteq I$ be the set of $i\in I$ with $|F_i|>N/20$. Assume for a contradiction that we have $|I^*| > |I|/2$. Then we have
\[
    \sum_{i \in I}|F_i| \geq \sum_{i \in I^*}|F_i| \\
    \geq \frac{N |I|}{40} \\
    \geq 2.5N\log(N),
\]
but this contradicts \Cref{lem:points-edges-hypercube}, which states that  $\sum_{i \in I} |F_i| \leq N\log N.$
\end{proof}

Using \Cref{lem:lot_of_bad_in_B}, we can prove that if our estimates $\bp_i$ were good, we reject with high probability if we have too many $*$ in $\tilde{\bd}$:

\begin{lemma}\label{lem:CheckStarsWithAssumption}

Assume that at the beginning of Phase~2, we have  
$$\left | \bp_i - \Prx_{\bz \sim f^{-1}(1)}[\bz_i=1]\right| \leq 0.15,\quad\text{for every $i \in [n]$.}$$
Then the probability of Phase~2 returning a $\tilde{\bd}$ with more than $100\log(N)$ $*$'s is at most $\delta/24$. 
\end{lemma}
\begin{proof}
For every $i\in \unfixed(\tilde{\bd})$ we have $\bp_i \in [0.25, 0.75]$. So, by assumption, for every $i \in {\unfixed(\tilde{\bd})}$, we must have $$\Prx_{\bz \sim f^{-1}(1)}[\bz_i=1] \in [0.1, 0.9].$$

Assume at the beginning of Phase~2 that the set $\unfixed(\tilde{\bd})$ has $|\unfixed(\tilde{\bd})| \geq 100 \log N$; otherwise the desired conclusion clearly holds.
Then if we sample $\bi \sim \unfixed(\tilde{\bd})$ uniformly at random, \Cref{lem:lot_of_bad_in_B} tells us that with probability at least $1/2$  we have that $|F_{\bi}| \leq N/20$.
    Fix such an $\bi$ and consider the event that $\bx,\by\sim \SAMP(f)$ satisfy
    $\bx_{\bi}=0$, $\by_{\bi}=1$ and 
    $f(\bx^{(\bi)})=f(\by^{(\bi)})=0$.

Focusing on the event on $\bx$ first, from above, $\bx \sim \SAMP(f)$ has $\bx_i=0$ with probability at least $0.1$. However by our bound on $|F_{\bi}|$, there are at most $N/20$ pairs $(z,z^{(i)})$ with both end points in $f^{-1}(1)$. So whenever we sample $\bx \sim f^{-1}(1)$ in step~14, with probability at least $1/20$ we have $\bx_{\bi}=0$ and $f(\bx^{(\bi)})=0$.

A similar argument holds for $\by$.
In summary, each loop between lines 12 and 16 rejects with probability at least $(1/2)\cdot (1/20)^2$.
As a result, given that we repeat the loop $\Omega(\log(1/\delta))$ times (with a large enough constant), it rejects with probability at least $1-\delta/24$.
\end{proof}

\subsubsection{Putting it all together, the proof of \Cref{lem:maintechnical}}
\label{sec:puttogether}

In this section, we prove, through a series of lemmas, that each point of \Cref{lem:maintechnical} holds. 
\begin{claim}[Item~(1.) of \Cref{lem:maintechnical}]\label{lem:IterativeBiasNoReject}
  If $f$ is unate, \IterativeBias\ never rejects.
\end{claim}
\begin{proof}
    During Phase~1, \IterativeBias\  rejects only if \CheckSample\ rejects. By \Cref{lem:CheckSamplesNoRej} this can't happen since $f$ is unate and $a \in f^{-1}(1)$. During Phase~2, \IterativeBias\  rejects only if on line~$15$ it found both a $1$-monotone and a $0$-monotone edge. Since $f$ is unate this can never happen. 
\end{proof}

We now turn to the next three items of \Cref{lem:maintechnical}: 
\begin{lemma} Items~(2.), (3.) and (4.) of \Cref{lem:maintechnical} all hold.
\end{lemma}
\begin{proof}
    Assume that during the $\ell$-iteration of the for-loop in Phase~1, $\bS$ or $\bT$ contains a point $\bx \in f^{-1}(1)$ with $f_{2N,a}(\bx)=0$, meaning $\absdist(\bx,a) > 2N$. 
    Then, by \Cref{lem:CheckSamplesNoRej}, the algorithm rejects with probability at least $1-\frac{\delta}{48\ell}$ at the beginning of Phase~1.2 Hence, by a union bound over all iterations of the for loop, it follows that probability that the algorithm fails to rejects but  $\bS \cup \bT \not \subseteq f_{2N,a}^{-1}(1)$ is at most:

    $$1-\sum_{k = 0}^{\infty} \frac{\delta}{48} \cdot \frac{1}{2^k} \geq 1-\delta/24.$$

    We now assume $\bS, \bT$ never contain a point $\bx$ with $f_{2N,a}(\bx)=0$. Since $f_{2N,a}^{-1}(1) \subseteq f^{-1}(1)$, this implies that the elements of $\bS,\bT$ are always distributed as i.i.d.~uniformly random samples from $f_{2N,a}^{-1}(1)$. 
    It thus follows from  \Cref{lem:IterativeBiasWithAssumption2} that:
   \begin{itemize}
        \item (Item~(2.) of \Cref{lem:maintechnical})  For any function $f$, the probability \IterativeBias\ accepts $f$ is at most $1/20+\delta/24 \leq 1/10$. 
    \end{itemize}

    Next, by \Cref{lem:IterativeBiasWithAssumption}, we have that with probability at most $\delta/12$, the algorithm reaches Phase~2 and there exists $i \in [n]$ violating \Cref{eq:close_p}. Assuming $\reldist(f,f_{2N,a}) \leq 1/20,$ by \Cref{cor:consistent} we have that:
    \begin{itemize}
        \item(Item~(3.) of \Cref{lem:maintechnical})  For any function $f, a \in \zo^n$ and $ \delta > 0$ if $\rd(f,f_{2N,a}) \leq 1/20$, the probability that \IterativeBias$(\MQ(h), \SAMP(h), a, \delta)$ returns $\tilde{\bd}$ which is not consistent with $d^f$ is at most $\delta/12+\delta/24 \leq \delta/6$.
    \end{itemize}
    Finally, assuming that the algorithm reaches Phase~2 and there are no $i \in [n]$ violating \Cref{eq:close_p}, by \Cref{clm:ananas} we have that for every $i \in [n]$, at the beginning of Phase~2 we have: $$|\bp_i - \Prx_{\bz \sim f^{-1}(1)}[\bz_i=1]| \leq 0.15 \quad {\text{for all~}i \in [n]}.$$
    Using \Cref{lem:CheckStarsWithAssumption} we can conclude that: 
    \begin{itemize}
        \item (Item~(4). of \Cref{lem:maintechnical})   For any function $f, a \in \zo^n$ and $ \delta > 0$, if $\rd(f,f_{2N,a}) \leq 1/20$ then the probability that \IterativeBias$(\MQ(h), \SAMP(h), a, \delta)$ returns $\tilde{\bd}$ with
    \begin{equation}
    \big|\{i\in [n]:\tilde{\bd}_i=*\}\big|\ge 100  \log N \tag{\ref{eq:smallstar}}
    \end{equation}
    is at most $\delta/24+\delta/12+\delta/24 = \delta/6$. \qedhere
    \end{itemize}
\end{proof}

\def\CheckSample{\hyperref[algo: CheckSamples]{CheckSamples}}

\begin{algorithm}[t!]
\caption{IterativeBias}
 \textbf{Input: }$\MQ(f)$ and $\SAMP(f)$ of some function $f$, $a \in f^{-1}(1)$, and a parameter $\delta$. \\
 \textbf{Output: } Either reject or a vector $\tilde{\bd}\in \{0,1,*\}^n$.\\
 \begin{tikzpicture}
\draw [thick,dash dot] (0,1) -- (16.5,1);
\end{tikzpicture}
\begin{algorithmic}[1]\label{algo: IterativeBias}
\Algphase{Phase 1: Estimate the bias of every coordinate}
    \For{each $\ell=1,2,4,8,\ldots, 10 \log n$} 
    \Algphase{\hspace{1cm} Phase 1.1:}
    \State Draw two sets $\bS,\bT$ of $k=4000(2\ell + \lceil\log(1/\delta)\rceil)$ points independently from $\SAMP(f)$.
    \State 
    For each $i\in [n]$, let 
    $\bp_i :=\Prx_{\bz \sim \bS}[\bz_i=1]$ and 
  $\bp_i' :=\Prx_{\bz \sim \bT}[\bz_i=1]$. 
   \vspace{0.2cm}
   \Algphase{\hspace{1cm} Phase 1.2:}
   \State Run \hyperref[algo: CheckSamples]{CheckSamples}$(f,\ba, \bS \cup \bT, {\delta}/({48\ell}))$ and reject if it rejects.
     \If{$|\bp_i-\bp_i'|\le 1/20$ for all $i\in [n]$}
    \State Exit the loop and go to \textbf{Phase~2}.  
    \ElsIf{$\ell=10\log n$}
    \State Halt and accept $f$.
    \EndIf
\EndFor\vspace{0.2cm}

\Algphase{Phase 2: Obtain a $\tilde{\bd}$ with at most $100\log(N)$ stars}
\State Let $\smash{{\tilde{\bd}} \in \zos^n}$ be defined as follows: 
    $$\tilde{\bd}_i=\begin{cases}
         1 & \text{if $\bp_i>0.75$}\\
         0 & \text{if $\bp_i<0.25$}\\
         * & \text{otherwise}
     \end{cases}.$$
     
    \For{$\ell=1$ to $4000\lceil\log (1/\delta)\rceil$}
    \State Draw  $\bi \sim \unfixed(\tilde{\bd}) $ uniformly at random.
    \State Draw two points $\bx,\by\sim \SAMP(f)$.
\State Halt and reject if  $\bx_{\bi}\ne \by_{\bi}$ and $f(\bx^{(\bi)})=f(\by^{(\bi)})=0$. \Comment{$\Edge^0_i(f), \Edge^1_i(f)$ are both $ \neq \emptyset$. }
\EndFor
\State Return $\tilde{\bd}$

\end{algorithmic}
\end{algorithm}

Finally, we prove point 5 of \Cref{lem:maintechnical}:

\begin{claim}
[Item (5.) of \Cref{lem:maintechnical}] 
For any function $f$, with probability at least $1-\delta/6$, \IterativeBias$(\MQ(f), \SAMP(f), a, \delta)$ makes at most $\tilde{O}(\log (N/\delta))$ queries to $\MQ(f)$ and $\SAMP(f)$.
\end{claim}
\begin{proof}
 First, note that each call to \CheckSample\ in Phase~1.2 uses $O(\log(\ell/\delta))$ queries by \Cref{lem:CheckSamplesNoRej}, and that Phase~2 uses $O(\log(1/\delta))$ queries. Let $\ell^*=2^j$ be the smallest power of $2$ that is at least $\log(N)$.
Now, if \IterativeBias\ stops at or before round $\ell^*$, then the total number of queries that it makes is at most
$$
O\left(\sum_{i=0}^j 2^i+ \log\left({\frac 1 \delta}\right) + \log\left(\frac{2^i}{\delta}\right)\right)=O\left(2^j+j \cdot \log\left({\frac 1 \delta}\right) +j \cdot \log\left(\frac{1}{\delta}\right)\right)=\tilde{O}\left(\log\left({N }\right)+\log\left({\frac 1 \delta} \right)\right).
$$
So it suffices to show that when \IterativeBias\ reaches $\ell^*$ in the for-loop, it halts with probability at least $1-\delta/6$.

To this end, we note that if $\bS$ or $\bT$ contains any point of distance more than $2N$ from $a$, it follows from \Cref{lem:CheckSamplesNoRej} that \IterativeBias\ halts and rejects in {Phase~1.2} with probability at least $1-\delta/(48\ell^*) \ge 1-\delta/12$.
So it suffice to show that \IterativeBias\ halts with probability at least $1-\delta/12$ conditioned on both $\bS$ and $\bT$ being drawn from $f_{2N,a}^{-1}(1)$.
Hence, in the rest of the argument, we suppose that both $\bS$ and $\bT$ are drawn from $f_{2N,a}^{-1}(1)$.
By \Cref{claim:trunc-fns-nontrivial-coord-bound}, 
we have that in $f_{2N,a}$ there are at most $2N^2$ coordinates with 
$$\Prx_{\bz \sim f_{2N,a}^{-1}(1)}[\bz_i = 1 ] \not \in \zo.$$ Let $\calT$ denote this set of at most $2N^2$ coordinates. 
Using ${\ell^*}\ge \log(N)$,
by a standard Chernoff bound (and choosing a sufficiently large constant in $k=\Omega(\ell + \log(1/\delta))$), we have that for any $i \in \calT$,
  $|\bp_i-\bp_i'|\le 1/20$ with probability at least $1-{(\delta/24N^2)}$. So by a union bound over all $2N^2$ coordinates in $\calT$ it follows that with probability at least $1-\delta/12$ when $\bS$ and $\bT$ are drawn from $f^{-1}_{2N,a}(1)$, either \CheckSample\ will pass the check of 
 {lines~5-6} of \IterativeBias\ and the algorithm will
  and moves to Phase~2, or it halts and rejects because of the call to \CheckSample\ on 
{line~4}. 
  If the algorithm reaches Phase~2, it will halt after $O(\log(1/\delta))$ additional queries. 
\end{proof}

\subsection{\hyperref[alg: Preprocessing]{Preprocessing}: 
Proof of \Cref{lem:preprocessing}}\label{sec:preprocessing}

The \hyperref[alg: Preprocessing]{Preprocessing} algorithm is given in \Cref{alg: Preprocessing}. 

We first prove the first item of \Cref{lem:preprocessing}:

\begin{claim}
[Item~(1.) of \Cref{lem:preprocessing}] 
 If $f$ is unate \hyperref[alg: Preprocessing]{Preprocessing} never rejects.
\end{claim}
\begin{proof}
    The algorithm can only reject if some call to \hyperref[algo: ConfirmDirection]{ConfirmDirection} returns $\Yes$. Note that by \Cref{lem:ConfirmDirection}, \hyperref[algo: ConfirmDirection]{ConfirmDirection}($\MQ(f), \SAMP(f), i, b)$ returns $\Yes$ if and only if the algorithm found an edge in $\Edge_{i}^b(f)$.

    Now, observe that during Phase~2, the algorithm calls \hyperref[algo: ConfirmDirection]{ConfirmDirection} on direction $\bj$ if and only if $\set{\bz^*, (\bz^*)^{{(\bj)}}} \in \Edge_{\bj}^{1-\tilde{d_{\bj}}}(f)$. But if $\Edge_{\bj}^{1-\tilde{d_{\bj}}}(f) \neq \emptyset$, since $f$ is unate, it must be that $\Edge_{\bj}^{\tilde{d_{\bj}}}(f) = \emptyset$ (see \Cref{cor:unate_Some_Empty_Edge}), and thus \hyperref[algo: ConfirmDirection]{ConfirmDirection} always returns $\No$. 

    Similarly, during Phase~3, the algorithm calls  \hyperref[algo: ConfirmDirection]{ConfirmDirection} on direction $i$ if and only if it found that $\{u, u^{(i)}\} \in \Edge^{1-\tilde{d}_i}_i(f)$. Since $f$ is unate, \hyperref[algo: ConfirmDirection]{ConfirmDirection} must always return $\No$.
\end{proof}

\begin{claim}
[Item~(2.) of \Cref{lem:preprocessing}] 
 For any $f$, the probability of \Preprocessing\
returning an $\bM$ with $\reldist(f,\smash{f_{\bM,\tilde{d}}})>\eps/10$ is at most $0.1$.
\end{claim}
\begin{proof}
   For any $M\geq 0$, we have that $\smash{f_{M,\tilde{d}}} \subseteq f^{-1}(1)$, and hence
\begin{equation}\label{eq:haha3}\rd(f,f_{M,\tilde{d}})=\frac{|f^{-1}(1)\setminus f_{M,\tilde{d}}^{-1}(1)|}{|f^{-1}(1)|}=\frac{\big | \{z \in f^{-1}(1) : \absdist(z, \tilde{d}) > M\} \big|}{|f^{-1}(1)|}.\end{equation}

   We can order points $x \in f^{-1}(1)$ based on 
   $\absdist(x, \tilde{d})$
   in decreasing order (breaking ties arbitrarily). Let $\tilde{z}$ be the $(\epsilon N/10)$-th {largest} element in such an ordering. 
   The probability that $\bS$ does not contain $\tilde{z}$ or any elements before it is at most 
    $$\big(1-(\epsilon/10)\big)^{30/\epsilon} \leq 0.1.$$
   On the other hand, if $\bS$ contains such a sample, the $\bM$ picked on line 1 would make sure that at most $\epsilon/10$ fraction of $x \in f^{-1}(1)$ have 
   $$\absdist(x, \tilde{d}) \leq \bM.$$ 
   In such a case we have by \Cref{eq:haha3} that $\rd(f,f_{M,\tilde{d}}) \leq \epsilon/10$.    
\end{proof}

\begin{claim}
[Item~(3.) of \Cref{lem:preprocessing}] 
If $\tilde{d}$ is consistent with $d^f$, then the probability of \hyperref[alg: Preprocessing]{Preprocessing} returning a value $\bM> 2\log N$ is at most $\delta/{6}$. 
\end{claim}
\begin{proof}
Assume that on line~$1$, the algorithm obtains $\bM > 2\log N$. This implies that $|\bT| > N^2$ and hence $|\bT| \geq 2N$. Using $|\bT| \geq 2N$, we now show that if Phase~2 does not already reject, the third phase rejects with high probability.

    Given that $N=|f^{-1}(1)|$, whenever we sample $\by \sim \bT$, with probability at least $1/2$ we have $f(\by)=0$.  Hence, the probability that for every $\by$ drawn on {line~10} we have {$f(\by)=1$} is at most $2^{-5\lceil\log(1/\delta)\rceil} \leq \delta/12.$
    
    Assume that some draw on line~10 gives a point $\by^*$ with $f(\by^*)=0$. After the binary search between $\bz^*$ and $\by^*$, the algorithm will find an $i \in [n]$ with $\tilde{d}_i \neq *$ on {line $15$}.  Since $\tilde{d}_i=d^f_i$, by \Cref{lem:ConfirmDirection}, we have \hyperref[algo: ConfirmDirection]{ConfirmDirection}($\MQ(f),\SAMP(f), i, \tilde{d}_i$), returns $\Yes$ with probability at least $0.9$. 
    
    Hence, the probability that some call to \hyperref[algo: ConfirmDirection]{ConfirmDirection}($\MQ(f),\SAMP(f), i, \tilde{d}_i$) returns $\Yes$ during {lines~16-18} is at least $1-(0.1)^{5\lceil\log(1/\delta)\rceil} \geq 1-\delta/12$. It follows by a union bound that the probability the algorithm rejects when $\bM > 2\log(N)$ is at least $1-2\cdot \delta/12 \geq 1-\delta/6$. 
\end{proof}

\begin{claim}[Item~(4.) of \Cref{lem:preprocessing}] 
   For any $f$ and $\tilde{d}$ that is consistent with $d^f$, the probability  that \hyperref[alg: Preprocessing]{Preprocessing}$(\MQ(f), \SAMP(f), \tilde{d},\eps,\delta)$ uses more than $$O((1/\eps)+ \log(N)+\log(1/\delta))$$ queries is at most $\delta/6$. 
\end{claim}
\begin{proof} 
Recall that each call to \hyperref[algo: ConfirmDirection]{ConfirmDirection} uses $O(1)$ queries (See \Cref{lem:ConfirmDirection}). 
So \hyperref[alg: Preprocessing]{Preprocessing} uses $O(1/\epsilon + \log(1/\delta))$ queries during its first two phases. 
In phase 3, the binary search between $\bz^*$ and $\by^*$ uses at most 
$$\log\Big(\big|\{j \mid \bz_j^* \neq \by_j^* \}\big|\Big)=O(\log(\bM)) \text{ queries}$$ queries. 
So Phase~3 uses $O(\log(1/\delta)+\log(\bM))$ queries. 

Hence, to prove the lemma, it suffices to show that if $\bM \geq 2N$, the algorithm reaches Phase~3 with probability at most $\delta/6$. Assume $\bM \geq 2N$ on line~1. 
Since $\bM \geq 2N$, we can have $f((\bz^*)^{(j)})=1$ for at most half the coordinates $j \in \left( \tilde{d} \Delta \bz^* \right)$. 
Hence, whenever we draw $\bj$ on line~$4$, with probability at least $1/2$ we have $f((\bz^*)^{(\bj)})=0$, which means that we run \hyperref[algo: ConfirmDirection]{ConfirmDirection}($\MQ(f),\SAMP(f), \bj, \tilde{d}_j)$. 
Since $\tilde{d}_{\bj}=d^f_{\bj}$ by \Cref{lem:ConfirmDirection}, ConfirmDirection returns $\Yes$ with probability at least $0.9$. 
Hence each iteration of lines 
{4 through 8}
rejects with probability at least $1/3$. 
Hence, the probability the algorithm doesn't reject during lines 
{3 through 9}
is at most $(1-1/3)^{10\lceil\log(1/\delta)\rceil}\leq \delta/6$.
\end{proof}

\begin{algorithm}[t!]
\caption{Preprocessing} \label{alg: Preprocessing}
 \textbf{Input: } $\MQ(f)$ and $\SAMP(f)$ of some function $f$,  a vector $\tilde{d} \in \{0,1,*\}^n$ and two parameters $\eps$ and $\delta$. \\
 \textbf{Output: } Either reject or a nonnegative integer $\bM$.\\
 \begin{tikzpicture}
\draw [thick,dash dot] (0,1) -- (16.5,1);
\end{tikzpicture}
\begin{algorithmic}[1]\label{algo: Preprossing}
\Algphase{Phase 1:}
    \State Draw $30\lceil1/\epsilon\rceil$ points $\bS$ from $\SAMP(f)$ and set 
    $$\bM :=\max_{\bz \in \bS}\left[\absdist(z, \tilde{d})\right].$$
\State Let $\bz^* \in \bS$ be a point that achieves the maximum above and let  
    $\bT \subseteq \{0,1\}^n$ be 
    $$\bT:=\big\{ y \in \zo^n:\bz^* \preceq_{\tilde{d}} y \big\},$$ 
where ``$\bz^* \preceq_{\tilde{d}} y$'' means $\bz^*_i \leq y_i$ for $i: \tilde{d}_i = 1$; $\bz^*_i \geq y_i$ for $i:\tilde{d}_i = 0$; $\bz_i^*=y_i$ for $i:\tilde{d}_i=*$.
\Algphase{Phase 2: Ensure $\bM \leq 2N$}
\RepeatN{$10\lceil\log(1/\delta)\rceil$}
\State Draw $\bj \sim z\absdist \tilde{d}$ and query $f((\bz^*)^{(\bj)})$.
\If{$f((\bz^*)^{(\bj)})=0$} \Comment{We have $\set{\bz^*, (\bz^*)^{(\bj)} \in \Edge_{\bj}^{1-\tilde{d}_{\bj}}(f)}$}
\State Run \hyperref[algo: ConfirmDirection]{ConfirmDirection}$(\MQ(f), \SAMP(f), \bj, \tilde{d}_{\bj})$. 
\State Halt and Reject if it returns $\Yes$. 
\EndIf
\End

\Algphase{Phase 3: Ensure $\bM  \leq 2\log(N)$}
\State Draw $5\lceil\log(1/\delta)\rceil$ many $\by \sim \bT$.
\If{$f(\by)=1$ for all drawn points}
\State Return $\bM$
\EndIf
\State Let $\by^*$ be the first point that was drawn with $f(\by^*)=0$.

        \State Perform Binary Search between $\bz^*$ and $\by^*$ to find an  
        
        $$ i\in \bz^*\Delta \by^* \text{ and $u$ such that } u_i \neq \tilde{d}_i, f(u)=1, f(u^{(i)})={0}. $$
\Comment{We have $\set{u,u^{(i)}} \in \Edge_{i}^{1-\tilde{d}_i}(f)$}
       \RepeatN{$5\lceil\log(1/\delta)\rceil$}
       \State Run \hyperref[algo: ConfirmDirection]{ConfirmDirection}$(\MQ(f), \SAMP(f), i, \tilde{d}_i)$;
       Halt and reject if it returns $\Yes$.
       \End
\State Return $\bM$.

\end{algorithmic} 
\end{algorithm}


\section{Lower bounds} \label{sec:lower}

\subsection{A two-sided non-adaptive lower bound: Proof of \Cref{thm:firstlowerbound}} 
\label{sec:mainfirstlowerbound}

Our proof is an adaptation of the $\tilde{\Omega}(\log N)$ lower bound from \cite{CDHLNSY2024}  for two-sided non-adaptive monotonicity testing in the relative-error model.  We will establish the following result:

\begin{restatable}{theorem}{nonadaptivelb}\label{thm:two-sided-non-adaptive-lb-hehe}
Let $\smash{N={n\choose 3n/4}}$. There is a constant $\eps_0>0$ such that any two-sided, non-adaptive algorithm for testing whether a function $f$ with $|f^{-1}(1)|=\Theta(N)$ is unate or has
  relative~distance~at least $\eps_0$ from unate functions must make $\tilde{\Omega}(\log N)$ queries.
\end{restatable} 

\cite{CDHLNSY2024} explains how, using a simple embedding argument, the monotonicity analogue of \Cref{thm:two-sided-non-adaptive-lb-hehe} implies the monotonicity analogue of \Cref{thm:firstlowerbound}.  Those arguments go through unchanged for unateness; we refer the reader to Remark~16 and the discussion following it in \cite{CDHLNSY2024} for details.
Thus, in what follows, we focus on establishing \Cref{thm:two-sided-non-adaptive-lb-hehe}.

\subsubsection{A useful class of functions for lower bounds: Two-layer functions} \label{sec:two-layer}

We say $f:\{0,1\}^n\rightarrow \{0,1\}$ is a \emph{two-layer} function if
\begin{equation} \label{eq:two-layer}
    f(x) = \begin{cases} 1 &\textrm{if } \Vert x \Vert_1 > 3n/4 + 1 \\ 
    0 &\textrm{if } \Vert x \Vert_1 < 3n/4 
    \end{cases}
\end{equation}
All functions used in our lower bound proofs are two-layer functions.

As is explained in \cite{CDHLNSY2024}, one reason why two-layer functions are helpful for lower bound arguments because the $\SAMP(f)$ oracle is not needed for two-layer testing --- it can be simulated, at low overhead, with the $\MQ$ oracle.  We refer the reader to Claim~15 of \cite{CDHLNSY2024} for details.

\subsubsection{Distributions $\Dyes$ and $\Dno$}\label{subsec:dist1}

The main difference between the current lower bound argument and the non-adaptive two-sided lower bound argument for relative-error monotonicity testing given in \cite{CDHLNSY2024} is in the $\Dn$ distribution.
To describe the construction of the yes- and no- distributions $\Dyes$ and $\Dno$, we begin by describing a distribution $\calE$ that plays an important role in both the $\Dyes$ and $\Dno$ constructions. 
$\calE$ is uniform over all tuples  $T = (T_i: i\in [L])$, where $L:=(4/3)^n$ and $T_i:[n]\rightarrow [n]$.
Equivalently, to draw a tuple $\TT\sim \calE$,
  for each $i\in [L]$, we sample a random $\TT_i$ by sampling 
  each $\bT_i(k)$ independently and
uniformly (with replacement) from $[n]$ for each $k\in [n]$.
We will refer to $T_i$ as the $i$-th term in $T$ and 
  $T_i(k)$ as the $k$-th variable of $T_i$.
Given a term $T_i:[n]\rightarrow [n]$, we abuse the notation to 
  use $T_i$ to denote the Boolean function over $\{0,1\}^n$ with
  $T_i(x)=1$ if $x_{T_i(k)}=1$ for all $k\in [n]$ and $T_i(x)=0$ otherwise.
  (So $T_i$ is a conjunction, which is why we refer to it as a term as mentioned above.)

A function $\ff \sim \Dy$ is drawn as described in \cite{CDHLNSY2024}. More explicitly, this is done using the following procedure:
\begin{flushleft}\begin{enumerate}
\item Sample $\TT\sim\calE$ and use it to define the multiplexer map $\bGamma = \bGamma_{\TT} \colon \{0, 1\}^n \to [L] \cup \{ 0^*, 1^* \}$:
$$
\bGamma_\TT(x) =\begin{cases}
0^* & T_i(x)=0\ \text{for all $i\in [L]$} \\
1^* & T_i(x)=1\ \text{for at least two different $i\in [L]$}\\
i & T_i(x)=1\ \text{for a unique $i\in [L]$}.
\end{cases}
$$ 
\item Sample  $\HH = (\hh_{i} \colon i \in [L])$ from $\Ey$, where each $\hh_{i} \colon \{0, 1\}^n \to \{0,1\}$ is independently (1) with probability $2/3$, a random dictatorship Boolean function, i.e., $\hh_{i}(x) = x_k$ with $k$ sampled  uniformly at random from $[n]$; and (2) with probability $1/3$, the all-$0$ 
    function.
\item Finally, with $\TT$ and $\HH$  in hand, $\ff =f_{\TT,\HH}\colon \{0, 1\}^n \to \{0, 1\}$ is defined as follows: $\ff(x) = 1$ if $\Vert x \Vert_1 > 3n/4 + 1$; $\ff(x) = 0$ if $\Vert x \Vert_1 < 3n/4 $; if $ \Vert x \Vert_1\in \{3n/4,3n/4+1\}$, we have
\[ 
\ff(x) = \begin{cases} 0 & \bGamma_{\TT}(x) = 0^* \\
						1 & \bGamma_{\TT}(x) = 1^* \\
						\hh_{\bGamma(x)}(x) & \text{otherwise (i.e., $\bGamma_{\TT}(x) \in [L]$)} \end{cases} 
\]
\end{enumerate}\end{flushleft}

While the $\Dy$ distribution is as in \cite{CDHLNSY2024}, the $\Dn$ distribution that we now describe is new.
A function $\ff \sim\Dn$ is drawn using the same procedure, with the exception that $\HH = (\hh_{i} \colon i \in [L])$ is drawn from $\En$ (instead of $\Ey$): Each $\hh_i:\{0,1\}^n\rightarrow \{0,1\}$ is (1) with probability $1/2$ a random dictatorship Boolean function $\hh_i(x) = x_{\bk}$ with $\bk$ drawn uniformly from $[n]$; and (2) with 
    probability $1/2$, a random anti-dictatorship Boolean function $\hh_i(x) = \overline{x_{\bk}}$ with $\bk$ drawn uniformly from $[n]$.  Note that every function $f$ in the support of either $\Dy$ or $\Dn$ is a two-layer function as defined in \Cref{eq:two-layer}.

The following lemma is proved in \cite{CDHLNSY2024} (see Lemma~17 of that paper):

\begin{lemma}\label{lem:mono}
Every function in the support of $\Dy$ is monotone (and hence unate).
\end{lemma}

We now show that a constant fraction of functions in $\Dn$ are constant-far from unate:
\begin{lemma}\label{lem:nonmono}
A function $\ff \sim \Dn$ satisfies $\reldist(\ff,\calC_{\textsc{unate}})=\Omega(1)$ with probability $\Omega(1)$. 
\end{lemma}
\begin{proof}
We start with some definitions. 

Fix a $T$.
For each $k\in [n]$ we write $X_k$ to denote the 
  following set of edges $(x,x^{(k)})$ such that 
\begin{enumerate}
\item $\|x\|_1=3n/4$ and $x_k=0$ (so $\|x^{(k)}\|_1=(3n/4)+1$);
\item $\Gamma_T(x)=\Gamma_T(x^{(k)})=i$ for some $i\in [L]$.
\end{enumerate}
In addition, after $H$ is fixed, we define $X_k^+$ to be
  the following subset of $X_k$:
\begin{enumerate}
\item $\|x\|_1=3n/4$  and $x_k=0$;
\item $\Gamma_T(x)=\Gamma_T(x^{(k)})=i$ for some $i\in [L]$;
\item $h_i(x)=x_k$ (and thus, $f_{T,H}(x)=0$ and 
  $f_{T,H}(x^{(k)})=1$).
\end{enumerate}
Similarly we define $X_k^-$ to be the following subset of $X_k$:
\begin{enumerate}
\item $\|x\|_1=3n/4$  and $x_k=0$;
\item $\Gamma_T(x)=\Gamma_T(x^{(k)})=i$ for some $i\in [L]$;
\item $h_i(x)=\overline{x_k}$ (and thus, $f_{T,H}(x)=1$ and 
  $f_{T,H}(x^{(k)})=0$).
\end{enumerate}
We note that after $T$ and $H$ are fixed, $\{X_k^+,X_k^-:k\in [n]\}$ are vertex-disjoint edges.
Given that $X_k^+$ contains monotone bichromatic edges and $X_k^-$ contains anti-monotone bichromatic edges, 
  at least
$$
\sum_{k\in [n]} \min\left(|X_k^+|,|X_k^-|\right)
$$
changes need to be made on $f_{T,H}$ to make it unate.
Given that trivially $|f_{T,H}^{-1}(1)|=O(N)$, it suffices to show with probability 
$\Omega(1)$ over $\bT\sim \calE$ and $\HH\sim \En$,
  we have 
$$
\sum_{k\in [n]} \min\left(|\bX_k^+|,|\bX_k^-|\right)
\ge \Omega(N).
$$

To this end, we start by proving the following claim about $\bT\sim \calE$:

\begin{claim}\label{claim:useagain}
With probability at least $\Omega(1)$, $\bT\sim \calE$ satisfies \begin{enumerate}
\item For every $i\in [L]$, the range of $\bT_i$ has size at least $\sqrt{n}$; and
\item We have 
$\sum_{k\in [n]}|\bX_k|\ge \Omega(nN).$
\end{enumerate}
\end{claim}
\begin{proof}
We work on the two bad events and apply a union bound.

For each $\bT_i$, the probability of $|\bT_i|\le \sqrt{n}$ is at most 
$$
\sum_{j\le \sqrt{n}} {n\choose j}\cdot \left(\frac{j}{n}\right)^n\le \sqrt{n}\cdot {n\choose \sqrt{n}}\cdot \left(\frac{1}{\sqrt{n}}\right)^n\le \exp\left( {-\Omega(n\log n)}\right).
$$
So the probability is $o_n(1)$ after applying a union bound over all $L$ terms.

It suffices to show that the second event happens with 
  probability at least $\Omega(1)$. To this end, it suffices to show that the expectation of $\sum_k |\bX_k|$ is at least $\Omega(nN)$ given that the sum is always bounded from above by $nN$. For this purpose, it suffices to work on any fixed edge $(x,x^{(k)})$ with
  $\|x\|_1=3n/4$ and $x_k=0$ and show that it is in $\bX_k$ with probability $\Omega(1)$ over $\bT\sim \calE$ and then apply linearity of expectation.

This probability is given by
$$
L\cdot \left(\frac{3}{4}\right)^n\cdot 
\left(1-\left(\frac{3n/4+1}{n}\right)^n\right)^{L-1}=\Omega(1),
$$
using the choice of $L=(4/3)^n$.
This finishes the proof of the claim.
\end{proof}

Fix a $T$ that satisfies the claim above.
We say an $X_k$, $k\in [n]$, is \emph{heavy} if $|X_k|=\Omega(N)$.
Using $\sum_k |X_k|\ge \Omega(nN)$ and the fact that every $|X_k|$ is trivially at most $N$, we have that the number of heavy $k\in [n]$ is at least $\Omega(n)$.
The lemma follows from the next claim about
  $\HH\sim \En$:

\begin{claim}
With probability at least $1-o_n(1)$, we have for every heavy $k\in [n]$:  $$|\bX_k^+|,|\bX_k^-|\ge \Omega(N/n).$$
\end{claim}
\begin{proof}
We focus on $\bX_k^+$; the case with $\bX_k^-$ is symmetric. 

 Consider the size of $\bX_k^+$
  as a function over $\bh_1,\ldots,\bh_L$ drawn from $\En$ one by one.
Then there are fixed nonnegative integers $m_1,\ldots,m_L$ such that $\sum_{i\in [L]} m_i=|X_k|=\Omega(N)$ and 
  $|\bX_k^+|$ 
  is a sum of independent random variables, where 
  each random variable is $m_i$ with probability $1/(2n)$
  and $0$ otherwise. 
Given that every  $T_i$ has at least $\sqrt{n}$ distinct variables, we have for every $i\in [L]$ that
$$
m_i\le {n-\sqrt{n}\choose 3n/4-\sqrt{n}}\le 2^{-\Omega(\sqrt{n})}\cdot N,
$$
and thus, 
$$
\sum_{i\in [L]} m_i^2 \le 2^{-\Omega(\sqrt{n})} N\cdot \sum_{i\in [L]} m_i=2^{-\Omega(\sqrt{n})} N\cdot |X_k|.
$$
Given that the  expectation of $|\bX_+^k|$ is $\mu:=|X_k|/2n=\Omega(N/n)$, we apply Hoeffding’s inequality:
$$
\Pr\left[\mu -|\bX_k^+|\ge \mu/2\right]\le \exp\left(-\Omega\left(\frac{\mu^2}{\sum_{i} m_i^2}\right)\right)=\exp\left(-2^{\Omega(\sqrt{n})}\right).
$$
The claim then follows from a union bound over all heavy $k\in [n]$ and both $\bX_k^+$ and $\bX_k^-$.
\end{proof}

The lemma follows directly from these two claims.
\end{proof}

\subsubsection{Sketch of the proof of \Cref{thm:two-sided-non-adaptive-lb-hehe} using \Cref{lem:mono,lem:nonmono}}

Our \Cref{lem:nonmono}, establishing that ``no-functions'' drawn from $\Dno$ have constant probability of being constant-far from unate, plays the role that Lemma~18 plays in \cite{CDHLNSY2024} (establishing that ``no-functions'' drawn from the $\Dno$ of that paper have constant probability of being constant-far from monotone).  The arguments which follow Lemma~18 of \cite{CDHLNSY2024} to establish the monotonicity lower bound apply, virtually unchanged, to establish the unateness lower bound in our context.
The only change which is needed is as follows:  in the proof of Claim~27 of \cite{CDHLNSY2024}, in place of Equation~8 of \cite{CDHLNSY2024} which states that
\[
\frac{\frac{2}{3}\cdot \frac{|A_{i,0}|}{n}+\frac{1}{3}}{\frac{2}{3}\cdot \frac{|A_{i,1}|}{n}}=
\frac{(2/3)\cdot |A_{i,0}|+(n/3)}{(2/3)\cdot |A_{i,1}|},   
\]
we now instead have that
\begin{equation} \label{eq:newratio}
\frac{\frac{1}{2}\cdot \frac{|A_{i,0}|}{n}+\frac{1}{2}\cdot \frac{|A_{i,1}|}{n}}{\frac{2}{3}\cdot \frac{|A_{i,1}|}{n}}=
\frac{(1/2)\cdot |A_{i,0}|+(1/2)\cdot |A_{i,1}|}{(2/3)\cdot |A_{i,1}|}.
\end{equation}
(The changes in the numerator above are easily seen to follow from the changes in how a draw of $\boldf \sim \Dno$ is generated in our context versus in \cite{CDHLNSY2024}.) 
The new expressions in the numerator propagate through to give corresponding new expressions in the analogues of Equations~(9) and (10) of \cite{CDHLNSY2024} in the obvious way, and that is the only change that is required to complete the proof of \Cref{thm:two-sided-non-adaptive-lb-hehe}.

\section*{Acknowledgements} \label{sec:ack}

D.P., K.P., and Y.Z.\ started working on this project for a course on Sublinear Algorithms taught by Sofya Raskhodnikova \cite{Raskhodnikova25}. Part of this work was carried out while visiting the Simons Institute for the Theory of Computing at UC Berkeley.

\begin{flushleft}
\bibliographystyle{alpha}
\bibliography{allrefs}
\end{flushleft}

\appendix


\section{Proofs of \Cref{theo:mono} and 
\Cref{thm:correctnessOfUnbiasedTest}}\label{appendix:testing}

\subsection{The proof of \Cref{theo:mono}}
\label{sec:monotonicity-in-G}

\begin{algorithm}[H]
\caption{BiasedTest}\label{algo: BiasedTest}
\vspace{0.15cm}
 \textbf{Input: }$\MQ(f)$ and $\SAMP(f)$ of some function $f$, a vector $\tilde{d} \in \{0,1,*\}^{n}$ and two parameters $M \in \mathbb{N}, \epsilon' \in (0,1)$.  \\ 
 \begin{tikzpicture}
\draw [thick,dash dot] (0,1) -- (16.5,1);
\end{tikzpicture}
\begin{algorithmic}[1]
\State If $\fixed(\tilde{d})=\emptyset$, return. 
    \RepeatN{ $10\lceil M/\epsilon'\rceil$ } 
    \State Draw $\bz \sim \SAMP(f)$. 
    \If{$\bz \Delta \tilde{d}=\emptyset$ (i.e., there is no $i \in \fixed(\tilde{d})$ with $\bz_i \neq \tilde{d}_i$)}
    \State Go to the next execution of the loop.
    \EndIf
    \State Uniformly sample $\bi \sim \bz \Delta \tilde{d}$ 
    to get an edge $\{\bz, \bz^{(\bi)}\}$.
    \If{$f(\bz^{(\bi)})=0$}  \Comment{$\set{\bz, \bz^{(\bi)}} \in \Edge_{\bi}^{1-\tilde{d_{\bi}}}(f)$}
    \State Run \hyperref[algo: ConfirmDirection]{ConfirmDirection}$(\MQ(f), \SAMP(f), \bi, \tilde{d}_{\bi})$. 
    \State Halt and reject $f$ if it returns $\Yes$. 
    \EndIf 
   
    \End
\State Return.
\end{algorithmic}
\end{algorithm}

\mono*

\begin{proof}
The bound on the query complexity is clear -- \hyperref[algo: BiasedTest]{BiasedTest} has $O(M/\varepsilon')$ iterations and makes constantly many queries in each iteration (the query complexity of \hyperref[algo: ConfirmDirection]{ConfirmDirection} is constant). 

We now prove correctness.
Note that the algorithm can only reject if it finds $\{\bz, \bz^{\bi}\} \in \Edge_{\bi}^{1-\tilde{d_{\bi}}}(f)$ in line~8, and then \hyperref[algo: ConfirmDirection]{ConfirmDirection}($\MQ(f),\SAMP(f), \bi, \tilde{d}_{\bi})$ returns $\Yes$ in line~9.
However by \Cref{lem:ConfirmDirection} the algorithm \hyperref[algo: ConfirmDirection]{ConfirmDirection} returns $\Yes$ if and only if it finds an edge in $\Edge_{\bi}^{\tilde{d_{\bi}}}(f)$. Therefore, the algorithm rejects only when it has found edges in both $\Edge_{\bi}^{1-\tilde{d_{\bi}}}(f)$ and $\Edge_{\bi}^{1-\tilde{1-d_{\bi}}}(f)$ for some $i\in [n]$ with $\tilde{d_i}\in \zo$, which violate the definition of unateness.

Next, we prove that the algorithm rejects with high probability if $\tilde{d}$ is consistent with $d^f$ and
\[
    \sum_{i \in \fixed(\tilde{d})}\left| \Edge^{1-\tilde{d}_i}_i \left( f_{M, \tilde{d}} \right) \right|\ge \eps' |f^{-1}(1)|.
\]

Since $\tilde{d}$ is consistent with $d^f$, for any $i \in \fixed(\tilde{d})$, we have $\tilde{d}_i = d^{f}_i$.
By \Cref{lem:ConfirmDirection}, \hyperref[algo: ConfirmDirection]{ConfirmDirection}$(\MQ(f), \SAMP(f), \bi, \tilde{d}_{\bi})$ returns $\Yes$ with probability at least $0.9$ for such $i$, in which case the algorithm rejects.
So it suffices to prove \hyperref[algo: ConfirmDirection]{ConfirmDirection} is called during some iteration with probability at least $\epsilon'/M$. 
Let $\mathcal{F} = \bigcup_{\substack{i \in \fixed(\tilde{d})}} \Edge^{1-\tilde{d}_i}_i \left( f_{M, \tilde{d}} \right)$.
By assumption, $|\calF| \geq \epsilon'|f^{-1}(1)|$.
Observe that the algorithm calls \hyperref[algo: ConfirmDirection]{ConfirmDirection} in line~9 if it found an edge $\set{\bz, \bz^{(\bi)}} \in \mathcal{F}$ in line~8.
Fix an arbitrary edge $e:=\{x,x^{(\ell)}\} \in \mathcal{F}$ with $f(x)=1$.
In line $4$ we have
\[
    \Pr[\bz =x] \geq 1/|f^{-1}(1)|.
\]
Since $x\in f^{-1}_{M, \tilde{d}}(1)$, there are at most $M$ coordinates $i\in [n]$ such that $i \in \fixed(\tilde{d})$ and $x_i \neq \tilde{d}_i$.
Thus, the probability that $\bi=\ell$ in line~7 is at least $1/M$ and the edge $e=\{x, x^{(\ell)}\}$ is found with probability at least $\frac{1}{M|f^{-1}(1)|}$. 
By summing over every edge in $\mathcal{F}$, for each iteration of the loop, the algorithm finds an edge in $\mathcal{F}$ with probability at least
\[
    \frac{|\mathcal{F}|}{M|f^{-1}(1)|} \geq \epsilon'/M.
\]
Then the probability the algorithms fails to reject is at most
\[
    (1-(\epsilon'/M)\cdot 0.9)^{10\lceil M/\epsilon'\rceil}\leq 0.01.
\]
\end{proof}

\begin{remark}
    As discussed in \Cref{subsubsection-nonadaptivity-known-N}, we could write \hyperref[algo: BiasedTest]{BiasedTest} to be non-adaptive. We present it like so for ease of exposition.
\end{remark}

\subsection{The proof of \Cref{thm:correctnessOfUnbiasedTest}}

\begin{algorithm}[H]
\caption{UnbiasedTest}
\vspace{0.15cm}
 \textbf{Input: }$\MQ(f)$ and and $\SAMP(f)$ of some function $f$, a vector $\tilde{d} \in \{0,1,*\}^{n}$ and parameter $\epsilon' \in (0,1)$.\\
 \begin{tikzpicture}
\draw [thick,dash dot] (0,1) -- (16.5,1);
\end{tikzpicture}
\begin{algorithmic}[1]\label{algo: UnbiasedTest}
\State If $ \unfixed(\tilde{d}) = \emptyset$, return. 
\For {$r=1, 2 \ldots, L:=\left\lceil\log\left( \frac{|\unfixed(\tilde{d})|}{\epsilon'} \right) \right\rceil+2$ }
\RepeatN{$s_r:= \left\lceil\frac{50 |\unfixed(\tilde{d})|}{\epsilon' 2^r}\right\rceil$} 
\State Sample $\bi \sim \unfixed(\tilde{d})$ uniformly at random.
\State Draw a set $\bS$ of $3\times 2^r$ samples from $\SAMP(f)$.
\For{each sample $\bz \in \bS$} 
\State Query $f(\bz^{(\bi)})$.
\EndFor 
\If{there are two samples $\bx,\by \in \bS$ with $\bx_{\bi}=0, \by_{\bi}=1$, and $f(\bx^{(\bi)})=f(\by^{(\bi)})=0$}
\State Halt and reject $f$. \Comment{\tiny $(\bx,\bx^{\bi})$ is monotone for $f$, $(\by,\by^{\bi})$ is anti-monotone for $f$. So $f$ is not unate. \normalsize}
\EndIf 
\End
\EndFor 
\State Return.
\end{algorithmic} 
\end{algorithm}

\unatethm*

We first bound the query complexity of the algorithm.
\begin{lemma}\label{lem:QC_Unateness_test}
    \hyperref[algo: UnbiasedTest]{UnbiasedTest}$(\MQ(f), \SAMP(f), \epsilon', \tilde{d})$ makes $\tilde{O}\left(\frac{|\unfixed(\tilde{d})|}{\epsilon'}\right)$ queries to $\MQ(f)$ and $\SAMP(f)$.
\end{lemma}
\begin{proof}
    Recall that $L$ satisfies $L=O\left(\log\left(\frac{|\unfixed(\tilde{d})|}{\epsilon'}\right)\right)$. Each iteration of lines $4-11$ uses $3\times 2^r$ queries to both $\SAMP(f)$ and $\MQ(f)$. Hence, the total number of queries is:

    $$\sum_{r=1}^L s_r \cdot 2\cdot (3\times 2^r) = \sum_{r=1}^L \frac{300|\unfixed(\tilde{d})|}{\epsilon'} = \tilde{O}\left(\frac{|\unfixed(\tilde{d})|}{\epsilon'}\right).$$
\end{proof}

Before giving the proof of the theorem we will need the following. 
\begin{lemma}\label{lem:viol_draw}
  Let $f:\zo^n\to \zo$, $i \in [n]$ and $b \in \zo$. We have:
   
   $$\Prx_{\bz \sim f^{-1}(1)}[ \set{\bz, \bz^{(i)}}  \in \Edge_{i}^b(f)] \geq \frac{|\Edge_{i}^b(f)|}{|f^{-1}(1)|}.$$ 
\end{lemma}
\begin{proof}
   Let $i \in [n]$. Without loss of generality assume $b=0$ (the proof for $b=1$ follows a similar argument). Let $\{x,x^{(i)}\} \in \Edge_i^1{(f)}$, where $x_i=0$, $f(x)=1$ and $f(x^{(i)})=1$ and $f(x)=0$. When we sample $\bz \sim \SAMP(f)$ with probability $1/(2|f^{-1}(1)|)$ we have $\bz=x$.  Hence, it follows that $$\Prx_{\bz \sim f^{-1}(1)}[ \Edge_{i}^1(f)] \geq \frac{|\Edge_{i}^1(f)|}{|f^{-1}(1)|}.$$ 
\end{proof}

The proof of  \Cref{thm:correctnessOfUnbiasedTest} follows the one of \cite{BCPRS20} using Levin’s ``work investment strategy''.
\begin{proof}[Proof of \Cref{thm:correctnessOfUnbiasedTest}]
    The bound on the query complexity of \hyperref[algo: UnbiasedTest]{UnbiasedTest} follows from \Cref{lem:QC_Unateness_test}. 
    The algorithm rejects if during some iteration of the two outer loops, it samples $i \in \unfixed(\tilde{d})$ and finds an edge in $\Edge_{i}^0(f)$ and an edge in $\Edge_{i}^{1}(f)$. 
    For $i \in \unfixed(\tilde{d})$ we let 
    
    $$\mu_i := \min\left(\left|\Edge_{i}^{0}(f)\right|
    , \left|\Edge_{i}^{1}(f)\right|\right)/(2|f^{-1}(1)|).$$
    
    For each $r\geq 1$, we denote by 
    
    $$S_r:=\{ i \in \unfixed(\tilde{d}) : 2^{-r}< \mu_i \leq 2^{-r+1} \}.$$

    Observe that $$\sum_{r\geq 1} |S_r|/2^r\geq \sum_{r\geq 1} \sum_{i \in S_r} \mu_i/2 \geq \sum_{i \in [n]} \min\left(\left|\Edge_{i}^{0}(f)\right|
    , \left|\Edge_{i}^{1}(f)\right|\right)/(4|f^{-1}(1)|) \geq \eps'/4.$$

    We also have that by setting $L:=\left\lceil\log\left(\frac{\unfixed(\tilde{d})}{\epsilon'}\right)\right\rceil+2$,
    
    $$\sum_{r > L} |S_r|/2^r \leq \frac{|\unfixed(\tilde{d})|}{2^{L+1}} \leq \epsilon'/8.$$
    
    As such, \begin{equation}\label{eq:lower_bound_SR}
        \sum_{r=1}^L |S_r|/2^r \geq \epsilon'/8.
    \end{equation}
    
    For any $r\geq 1$, let $p_r$ be the probability that one iteration of lines $4-11$ rejects. The probability that \hyperref[algo: UnbiasedTest]{UnbiasedTest} rejects is:

    $$1-\prod_{i=1}^L(1-p_r)^{s_r} \geq 1-e^{-\sum_{i=1}^L p_r s_r}.$$

    Fix $r\geq 1$ and $i\in \unfixed(\tilde{d})$ drawn on line $4$. 
    The event that lines $5-11$ reject for the chosen $r$ and $i$ means we get at least one edge in $\Edge_{\bi}^{0}(f)$ and one in $\Edge_{\bi}^{1}(f)$. 
    Using \Cref{lem:viol_draw}, we have that the probability that neither of these events happen is at most $2(1-\mu_i)^{3 \cdot 2^r}$. 
    Hence if $i \in S_r$, the algorithm rejects with probability at least $1-2(1 - 2^{-r})^{3\cdot 2^r} \geq 0.9$. 
    Therefore, we have that 
    
    $$p_r \geq 0.9 \cdot \frac{|S_r|}{|\unfixed(\tilde{d})|}.$$
    
    Now since $s_r \geq \frac{50|\unfixed(\tilde{d})|}{\eps' 2^r}$, using \Cref{eq:lower_bound_SR} we have:

    $$ \sum_{i=1}^{L} p_r s_r \geq (40 /\eps') \cdot \sum_{i=1}^{L}\frac{|S_r|}{2^r} \geq 5.$$

    Thus the algorithm rejects with probability at least $1-e^{-5} \geq 0.99$.
\end{proof}

\begin{remark}
    The algorithm \hyperref[algo: UnbiasedTest]{UnbiasedTest} itself is non-adaptive since the queries to $\MQ(f)$ (line 7) only depend on the samples from $\SAMP(f)$ and do not depend on the results of previous queries to $\MQ(f)$.
\end{remark}

\end{document}